\documentclass[a4paper,11pt]{article}
\pdfoutput=1
\usepackage{jheppub}
\makeatletter
\renewcommand{\@fpheader}{}
\makeatother
\usepackage[T1]{fontenc}
\usepackage{mathtools}
\usepackage{amsthm}
\renewcommand\beforetochook{\pagestyle{myplain}\pagenumbering{roman}\small}
\renewcommand\afterTocSpace{\normalsize\bigskip\medskip}
\renewcommand\afterTocRuleSpace{\clearpage}
\newcommand{\dd}{\mathrm{d}}
\newcommand{\Mminus}{\mathcal{M}_{-}}
\newcommand{\Mplus}{\mathcal{M}_{+}}
\newcommand{\Sh}{\Sigma}
\newcommand{\propacc}{a}
\newcommand{\DEC}{\textup{\textsc{dec}}}

\newtheorem{theorem}{Theorem}[section]
\newtheorem{proposition}[theorem]{Proposition}
\newtheorem{lemma}[theorem]{Lemma}
\newtheorem{corollary}[theorem]{Corollary}
\title{Radiative steering of warp shells}
\author[a,b,c]{An T. Le}
\affiliation[a]{College of Engineering and Computer Sciences, VinUniversity, Hanoi, Vietnam}
\affiliation[b]{Center for Environmental Intelligence, VinUniversity, Hanoi, Vietnam}
\affiliation[c]{Intelligent Autonomous Systems, TU Darmstadt, Germany}
\emailAdd{an.lt@vinuni.edu.vn}
\emailAdd{an@robot-learning.de}
\hypersetup{pdftitle={Radiative steering of warp shells},
 pdfauthor={An T. Le},pdflang={en-US}}
\abstract{
We construct exact timelike junctions between compact flat cavities and
a Kinnersley photon-rocket exterior. The shells satisfy the strict surface
dominant energy condition and steer subluminally under nonnegative
radiation. Burns joined at static spheres give $m_f/m_i=e^{-3L}$ for
exterior velocity-space path length $L$; the supported cavity center
obeys a separate sharp bound. Self-similar shells require anisotropic
stress and have growing normal modes with the ambient geometry and
coefficients held fixed.
A separate elastic model admits local self-gravitating recoil control
with outward emission. At zero gravitational coupling, exact unstrained
turns yield sharp fuel bounds and travel limits at fixed duration.
Limits on angular concentration and inward emission reduce efficiency.
For smooth emission and positive inner-face density, stress-free recoil
requires photons to enter the cavity. Self-gravitating settling after a
turn remains open.
}

\begin{document}
\addtocontents{toc}{\protect\setcounter{tocdepth}{2}}
\maketitle

\section{Introduction}
\label{sec:intro}

A compact flat cavity can move within a curved spacetime, but its
geometry alone does not specify a propulsion mechanism. Alcubierre's
warp metric violates the null energy condition~\cite{alcubierre1994}.
The obstruction extends to the Nat\'ario class with unit lapse and flat
spatial slices~\cite{natario2002,santiago2022}. The null-energy
argument assumes an Eulerian worldline, orthogonal to those slices,
that crosses the deformation and returns to the asymptotic region.
A flat cavity alone imposes neither metric restriction.
Work on subluminal shells includes the
classification proposed in~\cite{bobrick2021}, the constant-velocity
shell of~\cite{fuchs2024}, and a member of Huey's noncompact planar
membrane family satisfying the surface dominant energy condition
(DEC)~\cite{membrane2023}. Energy conditions do not by
themselves establish material viability or invariant
motion~\cite{barzegar2026,buchert2026frackowiak}.

We use the term warp shell for a compact cavity bounded by gravitating
matter and steered by radiative recoil. The exterior is the Kinnersley
photon rocket~\cite{kinnersley1969}; an exact timelike junction replaces
its singular source by a flat interior. Anisotropic photon emission
changes the total momentum. The Bondi news vanishes, so there is no
gravitational-wave flux at infinity~\cite{bonnor1994,damour1995,dain1996,podolsky2011}.
This recoil mechanism also underlies radiation-powered relativistic
travel~\cite{fuzfa2019}.
An observer needs a force to follow the accelerating cavity center;
freely falling test observers in the interior are inertial.
Flatness removes tidal curvature inside the cavity but does not carry
those observers with its accelerating boundary.

Sections~\ref{sec:construction} and \ref{sec:existence} establish the
junctions and their sharp mass bounds. Fixed-axis burns begin and end
on spheres at rest in their respective local frames, allowing successive
changes of direction. The exterior momentum is measured in a Bondi
frame; the supported-center velocity is defined by the interior
embedding. They obey different bounds.
Section~\ref{sec:material} derives anisotropy and normal compression
growth for the reconstructed surface law on a prescribed geometry.
These restrictions agree with membrane
analyses~\cite{yang2023,mourao2025}. For a barotropic fluid shell, the
fully coupled study~\cite{pitre2026} finds growing even-parity modes
at all sampled multipoles $\ell\ge2$ and parameter values.
Its constitutive assumptions differ from those of the radiating shell
considered here.

Section~\ref{sec:bulkbody} introduces a separate finite-thickness elastic model with
local self-gravitating recoil control. Section~\ref{sec:rigidarrival}
derives its exact unstrained maneuvers at $G=0$, sharp fuel limits,
and the corresponding distance bound at fixed travel time.
The inner-face force balance shows that nonzero
stress-free recoil cannot keep the cavity photon-free. Stress can
transmit thrust while emission is directed outward, but global
self-gravitating evolution and settling after a turn remain to be
established.

\section{Geometry, stress-energy, and recoil}
\label{sec:construction}
\label{sec:definition}

We use $G_{ab}=8\pi T_{ab}$, $G=c=1$, zero cosmological constant,
signature $(-,+,+,+)$, and
$R^a{}_{bcd}=2\partial_{[c}\Gamma^a_{d]b}+2\Gamma^a_{e[c}\Gamma^e_{d]b}$.
A dot denotes differentiation with respect to the canonical rocket
retarded time $u$, which labels outgoing null hypersurfaces.
The spacetime consists of a flat interior $\Mminus$, a timelike shell $\Sh$, and a photon-rocket exterior $\Mplus$.
The shell has topology $I\times S^2$ for a time interval $I$ and is given in exterior coordinates by $r=R(u,q)$, where $q=\cos\vartheta$ and $\varphi$ is the azimuthal angle.
Fixed-axis burns are concatenated in section~\ref{sec:polygonal}.
The cavity contains only test observers; a massive payload is not modeled.
For an axisymmetric interior embedding $(T,R\sqrt{1-q^2}\cos\varphi,R\sqrt{1-q^2}\sin\varphi,z)$, we define the equatorial center by $\Gamma(u)=(T(u,0),0,0,z(u,0))$.
The matching determines this curve and its acceleration.

\subsection{Photon-rocket exterior}
\label{sec:bulk}

In angles measured from the instantaneous acceleration axis, the exterior metric is~\cite{kinnersley1969,podolsky2011}
\begin{align}
 \dd s_+^2={}&-\left(1-\frac{2m}{r}+2arq-a^2r^2(1-q^2)\right)\dd u^2
       -2\dd u\dd r\nonumber\\
 &+2ar^2\dd u\dd q
       +r^2\left(\frac{\dd q^2}{1-q^2}+(1-q^2)\dd\varphi^2\right).
 \label{eq:rocketmetric}
\end{align}
Here $m(u)>0$. The massless metric describes Minkowski spacetime in
coordinates based on an accelerated worldline, called the seed.
At $m=0$, $u$ and $a$ are its proper time and signed proper acceleration.
At $m>0$, the shell
excises that singular worldline; the interior embedding determines the
actual cavity-center acceleration.

Einstein's equations give
\begin{equation}
 T^+_{ab}=\frac{n^2(u,q)}{r^2}\ell_a\ell_b,
 \qquad \ell_a=-(\dd u)_a,
 \qquad 4\pi n^2=-\dot m-3ma q.
 \label{eq:axin}
\end{equation}
Here $\ell^a=(\partial_r)^a$ is future null; $n^2$ denotes a coefficient,
not an assumed square. This is outgoing null dust, whose energy density
is proportional to $n^2$. All classical pointwise energy conditions hold
exactly when $n^2\ge0$, or
\begin{equation}
 -\dot m\ge3m|a|.
 \label{eq:box3}
\end{equation}
For $a=0$ the metric is outgoing Vaidya; if $m$ is also constant, it is Schwarzschild in retarded coordinates.

The vanishing Bondi news follows directly in Robinson--Trautman coordinates.
Let $v^\mu(u)$ be the seed four-velocity in a fixed asymptotic Lorentz frame, $\hat\ell^\mu=(1,\hat n)$, and $V=-v\cdot\hat\ell>0$.
The angular metric is $r^2V^{-2}\dd\Omega^2$, and the Bondi shear obeys~\cite{dain1996}
\begin{equation}
 \partial_{u_B}\sigma^0=\frac{\eth^2 V}{V}=0.
 \label{eq:nonews}
\end{equation}
Here $u_B$ is Bondi retarded time and $\eth$ is the spin-raising derivative on the unit sphere.
The monopole and dipole harmonics of $V=v^0-\vec v\cdot\hat n$
are annihilated by $\eth^2$. Thus the flux is entirely matter radiation.
This property does not extend to general Robinson--Trautman
rockets~\cite{dain1996}; its cosmological-constant generalization is
analyzed in~\cite{fernandez2026}.

\subsection{Momentum balance and the mass bound}
\label{sec:control}\label{sec:nofreelunch}

The Bondi four-momentum measures the isolated system's remaining energy
and momentum on a cut of future null infinity. In a fixed Bondi frame,
its balance law is~\cite{bondi1962,sachs1962}
\begin{equation}
 \frac{\dd P_{\rm B}^\mu}{\dd u_B}
 =-\oint\left(\frac{|\mathcal N|^2}{4\pi}+\mathcal F\right)
                \hat\ell^\mu\,\dd\Omega.
 \label{eq:t1gr}
\end{equation}
The two terms describe gravitational-wave and matter flux. We normalize
the complex news $\mathcal N$ by this equation;
$\mathcal F=\lim_{r_B\to\infty}r_B^2T_{u_Bu_B}$ is the outgoing matter flux density,
with $r_B$ the Bondi areal radius.
The balance assumes the usual asymptotically flat Bondi expansion and radiative matter falloffs.
For $\mathcal F\ge0$, the radiated four-momentum is an integral of future null vectors with nonnegative weights, hence is future causal:
\begin{equation}
 -\frac{\dd P_{\rm B}^0}{\dd u_B}
 \ge\left|\frac{\dd\vec P_{\rm B}}{\dd u_B}\right|.
 \label{eq:universalbound}
\end{equation}
Nonzero integrable angular flux gives strict inequality: equality requires
a single emission direction. Zero flux leaves total Bondi momentum fixed.

The corresponding mass cost follows from the relativistic rocket
equation~\cite{henriques2012}. Unit future four-velocities form a
hyperbolic space: its distance is relative rapidity, and path length
counts accumulated changes in rapidity, including turns.
\begin{lemma}[Mass cost of a momentum path]
\label{lem:covariantcost}
Let $P(s)$ be a $C^1$ future timelike momentum on a finite interval in
one Minkowski vector space, with $F=-dP/ds$ future causal or zero. Set
$M=\sqrt{-P^2}$, $V=P/M$, $w=\log(M_i/M)$, and
$E=-V\cdot F$. For $E>0$ define
$\eta=\sqrt{(F-EV)^2}/E$, and set $\eta=0$ when $E=0$.
Then $0\le\eta\le1$ and
\begin{align}
 L&:=\int\sqrt{(dV/ds)^2}\,ds=\int\eta\,dw,
 &d&:=\operatorname{arcosh}(-V_i\cdot V_f),\nonumber\\
 \log\frac{M_i}{M_f}-d
 &=(L-d)+\int(1-\eta)\,dw\ge0.
 \label{eq:covariantcost}
\end{align}
In particular, $M_f/M_i\le e^{-L}\le e^{-d}$.
Equality in the first bound requires null nonzero flux almost everywhere
with respect to $dw$; equality in the second requires a geodesic velocity
path traversed without reversal, allowing coasts.
\end{lemma}
\begin{proof}
Differentiating $P=MV$ and $V^2=-1$ gives
$E=-dM/ds\ge0$ and $F-EV=-M\,dV/ds$.
Causality gives $|F-EV|\le E$; $E=0$ forces $F=0$.
Thus $\sqrt{(dV/ds)^2}=\eta\,dw/ds$.
In the initial rest frame write $V=(\cosh\chi,\sinh\chi\,\hat n)$.
The velocity-space metric is
$d\chi^2+\sinh^2\chi\,d\hat n^2$, so $L\ge\int|d\chi|\ge d$.
Integration proves \eqref{eq:covariantcost} and its equality conditions.
\end{proof}

Here $\eta$ is the magnitude of exhaust momentum per unit emitted energy
in the instantaneous rest frame. For Bondi momentum, $s=u_B$ and $V$ is the
normalized total momentum.
Equation~\eqref{eq:covariantcost} separates excess path length from
exhaust inefficiency; neither $V$ nor its length need describe the cavity center.

For the photon rocket, the momentum on a canonical cut $u=\mathrm{const}$
is $P_{\rm B}^\mu=mv^\mu$~\cite{dain1996}.
In the angular formulas below, $V=-v\cdot\hat\ell$ is again the scalar
defined before \eqref{eq:nonews}.
If $\dd\Omega_*$ is instantaneous rest-frame solid angle and $k^\mu$ is normalized by $-v\cdot k=1$, then
\begin{equation}
 \frac{\dd P_{\rm B}^\mu}{\dd u}
 =-\oint n^2 k^\mu\dd\Omega_*
 =-\oint\frac{n^2}{V^3}\hat\ell^\mu\dd\Omega,
 \qquad
 k^\mu=\frac{\hat\ell^\mu}{V},\quad
 \dd\Omega_* =\frac{\dd\Omega}{V^2}.
 \label{eq:bondiflux}
\end{equation}
The time relation is $\partial u_B/\partial u=V$ at fixed asymptotic angle.
In the fixed-frame integral, $n^2$ is evaluated at the corresponding aberrated rest-frame direction.
In the instantaneous rest frame, \eqref{eq:axin} gives
\begin{equation}
 \oint n^2\dd\Omega_*=-\dot m,
 \qquad -\oint n^2 q\dd\Omega_*=ma,
 \label{eq:moments}
\end{equation}
since $\oint q\dd\Omega_*=0$ and $\oint q^2\dd\Omega_*=4\pi/3$.
These moments reproduce $\dot P_{\rm B}^\mu=\dot m v^\mu+m\dot v^\mu$.

Integration of \eqref{eq:box3} gives
\begin{equation}
 \frac{m_f}{m_i}\le \exp\left[-3\int_{u_i}^{u_f}|a|\dd u\right]
 \le e^{-3d},
 \qquad d=\left|\int_{u_i}^{u_f}a\dd u\right|.
 \label{eq:tsiolkovsky}
\end{equation}
Equality requires $-\dot m=3m|a|$ and, for the endpoint bound, a fixed
acceleration sign. The factor $3$ is the dipole efficiency penalty:
during emission, $\eta=m|a|/(-\dot m)\le1/3$ in
Lemma~\ref{lem:covariantcost}.
It follows from the restricted angular emission; general photon
emission has efficiency at most $1$.
For an initially resting rocket, fixed-frame radiated energy is $m_i-m_f\cosh d$; it differs from the invariant mass depletion $m_i-m_f$.

\subsection{Junction conditions and the static shell}
\label{sec:junction}\label{sec:anchor}

Let $h_{ab}$ be the induced metric, $N$ the unit normal directed from cavity to exterior, and $K_{ab}=e_a^\mu e_b^\nu\nabla_\mu N_\nu$.
We use $[X]=X^+-X^-$, with the same normal orientation on both sides.
The junction equations are~\cite{israel1966,poisson2004}
\begin{align}
 [h_{ab}]&=0,&
 S_{ab}&=-\frac{1}{8\pi}([K_{ab}]-[K]h_{ab}),\label{eq:box4}\\
 D_b S^b{}_a&=-[T_{Na}],&
 S^{ab}\bar K_{ab}&=[T_{NN}],\qquad
 \bar K_{ab}=\tfrac12(K^+_{ab}+K^-_{ab}).\label{eq:mombal}
\end{align}
Here $D$ is the intrinsic derivative, $T_{Na}=T_{\mu\nu}N^\mu e_a^\nu$, and $T_{NN}=T_{\mu\nu}N^\mu N^\nu$.
The first relation matches the geometry measured on the shell; the
curvature jump fixes its surface stress. The last two relations are
the tangential energy-momentum and normal force balances, following
from Gauss--Codazzi. Material and emission laws remain to be specified.

The surface dominant energy condition (DEC) requires $-S^a{}_b w^b$ to be future causal for every future timelike tangent vector $w^a$.
For Type~I stress, which has a timelike energy eigenvector, this is equivalent to $\varepsilon\ge\max(|P_1|,|P_2|)$ in its rest frame~\cite{hawking1973}.
We call the condition strict when $\varepsilon>\max(|P_1|,|P_2|)$.

For a static sphere of radius $R$ with flat interior, put $x=2m/R$ and
$s=\sqrt{1-x}$, the shell's redshift factor. Let $\tau$ be shell proper
time and $A$ either angular direction, without summation below.
The mixed principal curvatures are $K^{+\tau}{}_{\tau}=m/(R^2s)$, $K^{+A}{}_A=s/R$, $K^{-\tau}{}_{\tau}=0$, and $K^{-A}{}_A=1/R$.
The surface energy density $\sigma_0$ and tangential pressure $p_0$ are
\begin{equation}
 \sigma_0=\frac{1-s}{4\pi R},\qquad
 p_0=\frac{(1-s)^2}{16\pi Rs},\qquad
 \sigma_0-p_0=\frac{(1-s)(5s-1)}{16\pi Rs}.
 \label{eq:decfactor}
\end{equation}

The density and pressure are positive, and strict surface DEC holds exactly for
$1/5<s<1$, or $0<x<24/25$, the familiar flat-interior compactness bound~\cite{horvatilijic2007,andreasson2008sharp}.
The normalized margin $8\pi R(\sigma_0-p_0)$ is uniformly positive
on compact subintervals and persists under small exact deformations.

\section{Exact accelerating junctions}
\label{sec:existence}

We construct a contracting family to analyze the shell stress, then
a family with static ends that allows successive burns to be joined.

\subsection{Self-similar contraction}
\label{sec:homothetic}

Let the shell's size decrease with its mass. Set
\begin{equation}
 \mathcal R(u)=R_0-\nu(u-u_0)>0,\qquad
 m(u)=\frac{x\mathcal R(u)}2,\qquad
 \propacc(u)=\frac{\lambda}{\mathcal R(u)},\qquad
 R(u,q)=\mathcal R(u)\chi(q).
 \label{eq:homolaw}
\end{equation}
Here $R_0>0$, $x$ is the reference compactness, $\nu$ is the contraction rate, and $\lambda$ is the signed dimensionless acceleration.
The actual angular compactness is $2m/R=x/\chi$.

\begin{theorem}[Exact self-similar accelerating junction]
\label{thm:homothetic}
For each $x\in(0,24/25)$ there is $\epsilon_x>0$ such that \eqref{eq:homolaw}
admits a smooth embedded timelike junction to a flat cavity whenever
$\nu\ge3|\lambda|$ and $|\lambda|+\nu<\epsilon_x$.
Its profile is analytic across $[-1,1]$, with $\chi(0)=1$ and $\chi'(0)=0$.
On every compact time interval with $\mathcal R>0$, the shell bounds a regular cavity,
the glued region is stably causal, the surface stress is Type~I with strict \DEC{},
and the exterior satisfies all classical bulk energy conditions.
For $\lambda\ne0$, both the exterior steering and the interior equatorial-center acceleration are nonzero.
\end{theorem}

Appendix~\ref{app:homothetic} gives the angular equation, embedding,
and uniform normalized DEC margin.

For a contracting segment ($\nu>0$) with dimensionless duration $t_f$, where $\dd t=\dd u/\mathcal R$,
\begin{equation}
 \frac{m_f}{m_i}=\frac{\mathcal R_f}{\mathcal R_i}=e^{-\nu t_f},
 \qquad d=|\lambda|t_f,
 \qquad u_f-u_i=\frac{\mathcal R_i}{\nu}(1-e^{-\nu t_f}).
 \label{eq:homobudget}
\end{equation}
Thus $\nu=3|\lambda|$ saturates the exterior mass bound.
This family contracts throughout, with no static endpoints.

\subsection{Smooth burns with static ends}
\label{sec:endpointburn}

Here each endpoint is static in its own rest frame. A burn changes the
velocity between those frames. Write $\chi_c$ for the supported
center's rapidity in the interior Minkowski frame.

\begin{theorem}[Smooth static-to-static acceleration]
\label{thm:endpointburn}
Let $R_0>0$, $0<2m_i/R_0<24/25$, $d>0$, and $\gamma\ge3$.
There is a smooth embedded axisymmetric timelike junction, defined for $u\in\mathbb R$,
between a photon-rocket exterior and a regular flat cavity, with
\[
 R(u,0)=R_0,\qquad R_q(u,0)=0,\qquad
 -\dot m=\gamma m\propacc,\qquad
 \propacc\ge0,\qquad \int\propacc\,\dd u=d.
\]
Outside a finite retarded-time interval, $\propacc=0$ and $R(u,q)=R_0$.
The endpoint masses are $m_i$ and $m_f=m_i e^{-\gamma d}$.
The first junction holds exactly, the surface stress is Type~I with strict
\DEC{}, the exterior null radiation is nonnegative, and the glued spacetime
is stably causal.
For the interior equatorial center,
\begin{equation}
 \Delta\chi_c=\frac2\gamma
   \left(\operatorname{artanh}s_f-\operatorname{artanh}s_i\right)>0,
 \qquad s_{i,f}=\sqrt{1-2m_{i,f}/R_0}.
 \label{eq:endpoint_center_gain}
\end{equation}
\end{theorem}

Appendix~\ref{app:endpointburn} constructs and concatenates small
pulses that join static intervals with all derivatives matching there.
The case $\gamma=3$ attains the dipole bound. The proof chooses the
duration and requires sufficiently small amplitudes; it gives no sharp
acceleration threshold.

\subsection{Direction changes and the cavity-center bound}
\label{sec:polygonal}

Let $\mathbb H^3$ be the future unit hyperboloid in the exterior seed Minkowski space, with distance $d_{\mathbb H}(v,w)=\operatorname{arcosh}(-v\cdot w)$.
Its geodesics are fixed-axis boosts.
We call the static intervals adjoining each burn collars.

\begin{theorem}[Exact steering by successive burns]
\label{thm:turningjunction}
Fix $R_0>0$, $0<2m_0/R_0<24/25$, and a finite sequence $v_0,\ldots,v_N\in\mathbb H^3$.
Set $d_j=d_{\mathbb H}(v_{j-1},v_j)$ and $L=\sum_{j=1}^N d_j$.
There is a smooth exact timelike cavity junction whose exterior velocity traverses the successive geodesic segments, with static spherical ends and intermediate collars of radius $R_0$, and masses
\begin{equation}
 m_j=m_0\exp\left(-3\sum_{\ell=1}^j d_\ell\right).
 \label{eq:polygonmass}
\end{equation}
The cavity is globally embedded and flat, the glued spacetime is stably causal, the surface has strict Type~I DEC, the null radiation is nonnegative, and Bondi news vanishes.
The pulse shapes and waiting times are chosen by the construction.
For $D=d_{\mathbb H}(v_0,v_N)$, $m_N/m_0=e^{-3L}\le e^{-3D}$; equality holds exactly when the nonconstant segments follow one geodesic without backtracking.
\end{theorem}

Rotational symmetry allows the next acceleration axis to be chosen
during a static collar. Spacelike cuts separate successive interior
pieces and prevent overlap (appendix~\ref{app:endpointburn}). This
realizes the classical photon-rocket budget~\cite{podolsky2011} with a
compact cavity. Axis changes during nonzero thrust are not covered.

\begin{corollary}[Sharp mass bound for the supported center]
\label{cor:centerbudget}
Fix $R_0>0$ and $0<2m_i/R_0<24/25$.
Let $C$ be the hyperbolic distance between the supported center's endpoint four-velocities in the interior Minkowski frame.
For a fixed-axis junction with $R(u,0)=R_0$, $R_q(u,0)=0$, a timelike center, static ends, and nonnegative radiation, or a concatenation in Theorem~\ref{thm:turningjunction},
\begin{equation}
 \frac{m_f}{m_i}\le
 \left[\cosh\frac{3C}{2}+s_i\sinh\frac{3C}{2}\right]^{-2},
 \qquad s_i=\sqrt{1-2m_i/R_0}.
 \label{eq:centerbudget}
\end{equation}
For every finite $C\ge0$, a straight saturated burn attains equality, with no prescribed duration.
\end{corollary}
\begin{proof}
For fixed-axis motion, the center connection \eqref{eq:endpoint_centerrate} and the static endpoints give $C=|\int a/s\,\dd u|$, where $s=\sqrt{1-2m/R_0}$.
Since $-\dot m\ge3m|a|$ and $\dot s=-\dot m/(R_0s)$,
\begin{equation}
 C\le\int\frac{|a|}{s}\,\dd u
 \le\frac23\left(\operatorname{artanh}s_f-\operatorname{artanh}s_i\right).
 \label{eq:centerintegral}
\end{equation}
For concatenated saturated burns, the interior velocity triangle inequality bounds $C$ by the sum of the increments \eqref{eq:endpoint_center_gain}; that sum telescopes to the same right-hand side.
Put $A=3C/2$.
Then $s_f\ge(s_i+\tanh A)/(1+s_i\tanh A)$, and $m_f/m_i=(1-s_f^2)/(1-s_i^2)$ gives \eqref{eq:centerbudget}.
Theorem~\ref{thm:endpointburn} with $\gamma=3$ attains both inequalities for any $C>0$; a static shell covers $C=0$.
\end{proof}

As $2m_i/R_0\to0$, the bound tends to $e^{-3C}$.
At a static endpoint the shell's proper energy is
\begin{equation}
 E_\Sigma=4\pi R_0^2\sigma_0=R_0(1-s),\qquad
 m=E_\Sigma-\frac{E_\Sigma^2}{2R_0}.
 \label{eq:bindingenergy}
\end{equation}
Thus even at rest the exterior mass includes gravitational binding.
The rapidity $C$ compares interior velocities, whereas $m$ normalizes the
exterior Bondi momentum. Lemma~\ref{lem:covariantcost} requires mass and velocity from the same
momentum, so it cannot be applied to $(m,C)$. Figure~\ref{fig:steering}
compares the exterior and supported-center bounds.

\begin{figure}[p]
\centering
\includegraphics[width=\textwidth,height=.75\textheight,keepaspectratio]{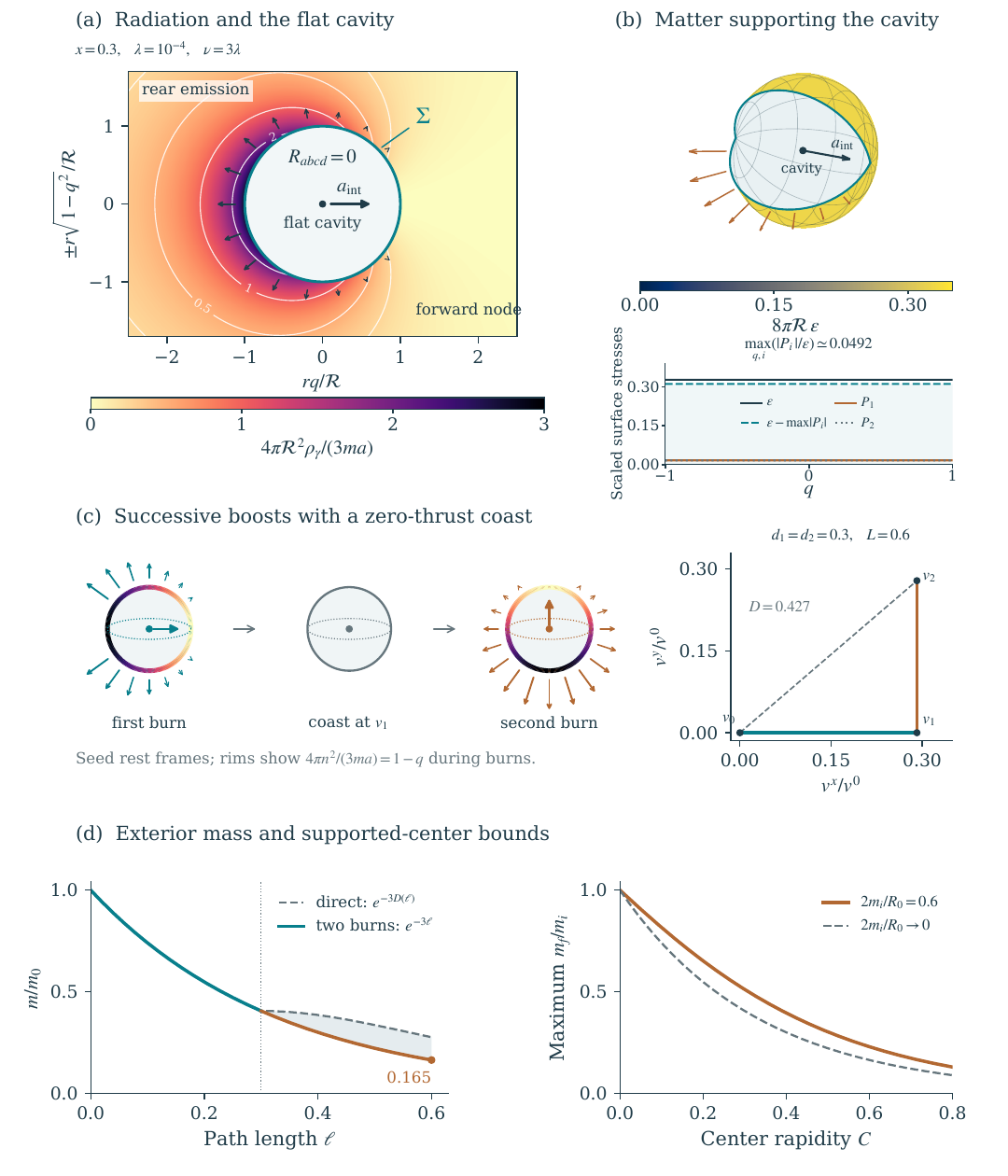}
\caption{Radiative steering and shell stress.
(a,b) Self-similar solution $x=0.3$, $\lambda=10^{-4}$, $\nu=3\lambda$,
with boundary from \eqref{eq:homofirst}.
The map shows $\rho_\gamma=T_{ab}U^aU^b=n^2/(r^2F)$ for
$U=F^{-1/2}\partial_u$, $F=-g_{uu}>0$; exterior arrows follow null rays.
Outer-face colors show energy-frame density; inset stresses and DEC
margin are scaled by $8\pi\mathcal R$. The pressure curves overlap at this scale.
These coordinate projections are not isometric spatial embeddings;
deformation is below line width. Interior arrows
denote supported-center acceleration; free test particles remain inertial.
(c) Perpendicular boosts in successive seed rest frames, changing axes
during a zero-thrust coast at $v_1$:
$d_1=d_2=0.3$, $L=0.6$, $D=\operatorname{arcosh}(\cosh^2 0.3)\simeq0.4274$.
Rims encode exterior dipole power $1-q$; velocity axes share a fixed
exterior frame. (d) Exterior mass cost versus a direct boost to the same
velocity, timing free, and the separate interior-center bound
\eqref{eq:centerbudget}. Burns in (c,d) have static ends; (a,b) contract.
DEC does not establish stability.}

\label{fig:steering}
\end{figure}

\section{Material restrictions and normal stability}
\label{sec:material}

\subsection{Anisotropy}

The junction determines the stress along a chosen geometry. A constitutive
law must also specify its response to strain and emission~\cite{mourao2025}.
In an orthonormal shell frame with a spatial leg along a meridian, write
$S_{\hat a\hat b}=\bigl(\begin{smallmatrix}\sigma&j&0\\j&p_1&0\\0&0&p_2\end{smallmatrix}\bigr)$.
Here $j$ is the meridional energy flux. The energy frame is the local
rest frame in which this flux vanishes.
Near the static Type~I \DEC{} branch, its energy-frame density and meridional pressure are
\[
 \varepsilon=\frac{\sigma-p_1+d_L}{2},\qquad
 P_1=\frac{p_1-\sigma+d_L}{2},\qquad
 d_L=\sqrt{(\sigma+p_1)^2-4j^2},\qquad P_2=p_2.
\]
Since $\varepsilon+p_2>0$ near the static state, the identity
\begin{equation}
 (P_1-p_2)(\varepsilon+p_2)
   =(\sigma+p_2)(p_1-p_2)-j^2=: \mathcal I
 \label{eq:fluidtest}
\end{equation}
shows that a perfect surface fluid exists in some frame if and only if $\mathcal I=0$.
For the self-similar family, analyticity of the profile and reflection symmetry give, at the equator,
\begin{equation}
 (8\pi\mathcal R)^2\mathcal I
 =\lambda^2\left[\frac{3(1-s)^2}{s^3}
        +O(|\lambda|+\nu)\right],\qquad s=\sqrt{1-x}.
 \label{eq:anisotropyexpansion}
\end{equation}
Appendix~\ref{app:homothetic} derives the positive leading coefficient.
For each fixed $0<x<24/25$, every sufficiently small member with $\lambda\ne0$ therefore requires anisotropic stress in its energy frame.

\subsection{Constitutive reconstruction of a finite history}
\label{sec:globalmaterial}

We fit a material law to a prescribed finite shell history and check
its intrinsic elastic wave speeds.

\begin{theorem}[Constitutive realization on a finite shell history]
\label{thm:globalmaterial}
On any compact slab of a strict Type~I DEC history in
section~\ref{sec:existence}, extended slightly at its ends, there is a
smooth surface law on $\mathcal B\simeq S^2$ reproducing its stress and
outgoing flux. In an open state neighborhood it has a conserved particle
current, strict DEC, nonnegative emission, and, on the prescribed geometry,
real strictly subluminal intrinsic elastic speeds. Reaction variables
are transported by the timelike material flow.
\end{theorem}

The labels $X^A$ identify material elements and remain constant along
the future unit material velocity $U$. The operator $U$ differentiates
along material proper time. The strain is $H^{AB}=h^{ab}D_aX^AD_bX^B$.
A fixed material area form defines the conserved particle density $n$, distinct from the radiation coefficient $n^2$ in \eqref{eq:axin}.
With an internal clock $U\tau=1$ and carried energy per particle $z>0$, take
\begin{equation}
 \mathcal E(\tau,X,H,z)
   =\frac12Q_{AB}(\tau,X)H^{AB}+V(\tau,X)+n(H,X)z.
 \label{eq:globalenergy}
\end{equation}
The coefficients $Q>0$ and $V$ fit the stress; $nz$ adds inertia without
stiffness. The acoustic Rayleigh quotient is
\begin{equation}
 c^2(v)=\frac{Q_{AB}v^Av^B}
 {Q_{AB}v^Av^B+nz(H^{-1})_{AB}v^Av^B},\qquad v\ne0,
 \label{eq:materialspeed}
\end{equation}
which lies strictly between zero and one.
For emitted power $\mathcal L=-\bar T^+_{Na}\bar U^a\ge0$, the laws
\begin{equation}
 U\tau=1,\qquad Uz=\frac{-\mathcal L-\mathcal E_\tau}{n}
 \label{eq:globalreaction}
\end{equation}
give energy balance including pressure work. Overbars here denote the
prescribed history. Appendix~\ref{app:material}
constructs the coefficients and flux. This law is tailored to the history;
normal stability requires a separate test.

\subsection{Short-wavelength normal perturbations}

Use the prescribed Gaussian normal metric
$g=\dd\ell^2+(h_{ab}+2\ell\bar K_{ab}+O(\ell^2))\dd x^a\dd x^b$,
where $\ell$ is signed proper distance from the shell and the background
obeys \eqref{eq:mombal}. We test displacements perpendicular to the shell
while keeping the ambient geometry fixed. Freezing coefficients at one
event then isolates the local short-wavelength behavior.

\begin{proposition}[Instability of the frozen surface equations]
\label{prop:membrane}
Consider the material, reaction, and emission laws of Theorem~\ref{thm:globalmaterial}, together with the normal balance, on a prescribed ambient geometry.
Assume no bending moments or additional normal traction.
Freeze the linearized coefficients in a material rest frame, and let
$p(\hat k)=P^{ij}\hat k_i\hat k_j$ for a unit spatial direction $\hat k$.
If $p(\hat k)>0$, the full frozen surface system has modes
$e^{\sigma(k)t+ik\hat k\cdot x}$ with nonzero normal displacement and
\begin{equation}
 \sigma(k)=k\sqrt{\frac{p(\hat k)}{\varepsilon}}+O(1),
 \qquad k\longrightarrow\infty.
 \label{eq:surfacegrowth}
\end{equation}
Intrinsic elastic moduli and local reaction and emission response do not remove this branch.
\end{proposition}
\begin{proof}
For normal displacement $\xi$, the quadratic action
$\tfrac12S^{ab}\partial_a\xi\partial_b\xi$ gives
\begin{equation}
 \varepsilon\partial_t^2\xi+P^{ij}\partial_i\partial_j\xi=0.
 \label{eq:normalprincipal}
\end{equation}
The variation $\delta h_{ab}=2\bar K_{ab}\xi$ couples only first
normal derivatives to the intrinsic equations;
$\delta S^{ab}\bar K_{ab}$ contains only first derivatives of
$\eta=\delta X$. Algebraic emission traction and local reactions add
no second displacement derivatives.
With $\alpha=(\delta\tau,\delta z)$ and
$W=(k\xi,\partial_t\xi,k\eta,\partial_t\eta,\alpha)$, the bounded
triangular reduction of appendix~\ref{app:material} gives
$W_t=[kA_1(\hat k)+O(1)]W$. Apart from zero transport roots, its
principal characteristic factors are
\begin{equation}
 (\varepsilon\lambda^2-p(\hat k))\det(\lambda^2M+Q),\qquad
 M=Q+nzH^{-1}>0.
 \label{eq:surfacefactors}
\end{equation}
The elastic roots are imaginary. The positive normal root
$c_*=\sqrt{p(\hat k)/\varepsilon}$ is simple and isolated, so analytic
perturbation in $1/k$ gives $\sigma=kc_*+O(1)$ and an eigenvector
converging to the normal eigenvector.

\end{proof}

The sign in \eqref{eq:normalprincipal} is the usual membrane buckling sign~\cite{yang2023,mourao2025}:
positive pressure is compression, whereas negative pressure is tension.
For the static Minkowski--Schwarzschild shell,
\begin{equation}
 c_*^2=\frac{p_0}{\sigma_0}=\frac{1-s}{4s}>0,
 \qquad s=\sqrt{1-x}\in(0,1).
 \label{eq:staticgrowth}
\end{equation}
Both exact families retain a uniform positive $c_*$ on sufficiently small
compact histories.

Appendix~\ref{app:normalresponse} proves the resulting frozen Sobolev
ill-posedness: arbitrarily small high-frequency data can produce
arbitrarily large Sobolev norms at a fixed positive time. It also gives a sufficient bulk-response
condition for coupled growth. That condition remains unproved for the
radiating junction.

A bare bending stiffness gives $\omega^2=(Bk^4-pk^2)/\varepsilon$, bounding the
unstable band but making $|\omega/k|$ unbounded. A causal restoring mechanism needs
additional dynamics; finite thickness also requires a new material equilibrium.

\section{Finite-thickness matter and radiative control}
\label{sec:bulkbody}

The need for a different wall model is already visible in the static
limit. Consider the spherical Kerr--Schild ansatz, written in static
coordinates as
\begin{equation}
 \dd s^2=-f(r)\dd t^2+f(r)^{-1}\dd r^2+r^2\dd\Omega^2,
 \qquad f(r)=1-\frac{2M(r)}r>0.
 \label{eq:radialtensionmetric}
\end{equation}
Let $M\in C^2([r_-,r_+])$, with $0<r_-<r_+$ and
$M(r_-)=M'(r_-)=0$, so the metric joins a flat cavity with continuous
first derivatives. Primes denote $r$ derivatives. In a static orthonormal
frame, Einstein's equations give the density and principal
pressures~\cite{balart2014}:
\begin{align}
 8\pi\rho&=\frac{1-f-rf'}{r^2}=\frac{2M'}{r^2},
 &p_r&=-\rho,\nonumber\\
 8\pi p_t&=\frac{f''}{2}+\frac{f'}r=-\frac{M''}r.
 \label{eq:radialtensionstress}
\end{align}
Hence $\rho+p_t=-r\rho'/2$. The weak energy condition requires
$\rho\ge0$ and $\rho+p_t\ge0$, so $\rho'\le0$.
Since $\rho(r_-)=0$, these inequalities force $\rho=0$ throughout
$[r_-,r_+]$, and then $M'=0$ and $M=0$.
Thus any nontrivial wall in this ansatz violates the weak energy
condition, and therefore DEC. The imposed radial tension $p_r=-\rho$
is essential to this obstruction.

We use a finite-thickness elastic solid with independent radial and
tangential stresses, its own material dynamics, and a self-consistent
exterior. It is a separate model; no limit connecting it to the exact
thin shell is established. Its static spherical reference has a flat
vacuum cavity by spherical symmetry; a general deformation need not.

\subsection{Constitutive law and static cavity bodies}
For a three-dimensional body, let $X^A$ label material particles.
Use a Euclidean reference metric and unit reference particle density.

Let $H^{AB}=g^{\mu\nu}\partial_\mu X^A\partial_\nu X^B>0$,
$n=\sqrt{\det H}$, and let $z\ge0$ be fuel energy per particle.
The future unit material velocity $U$ satisfies $UX^A=0$.
The state $H=I$ is the unstrained reference state.
For $m,\mu,\kappa>0$, take
\begin{equation}
 \rho(H,z)=n\left[m+z+\frac{\mu}{2}
       \left(\operatorname{tr}H^{-1}-3+\log\det H\right)
                     +\frac{\kappa}{8}(\log\det H)^2\right].
 \label{eq:bulkenergy}
\end{equation}
Write $W_0=\rho/n$, the energy per particle. This is a compressible
neo-Hookean elastic law with carried fuel energy.
Here the constant $m$ is rest energy per particle, distinct from the exterior mass $m(u)$.

\begin{proposition}[Local properties of the bulk constitutive law]
\label{prop:bulkcandidate}
At $H=I$, \eqref{eq:bulkenergy} has $\rho=m+z$, zero principal pressures,
and transverse and longitudinal elastic speeds
\begin{equation}
 c_T^2=\frac{\mu}{m+z},\qquad
 c_L^2=\frac{\kappa+2\mu}{m+z}.
 \label{eq:bulkspeeds}
\end{equation}
If $\kappa+2\mu<m$, its relaxed elastic kinetic and strain quadratic forms are positive
modulo rigid motions, and its elastic characteristics are strictly subluminal
on an open neighborhood of the relaxed states with $z$ in any fixed compact nonnegative interval.
Its stress satisfies strict DEC on such a neighborhood.
\end{proposition}
\begin{proof}
Let the principal stretches be $e^{r_i}$, put $\theta=\sum_i r_i$, and write $\rho=nW$.
Then
\[
 W=m+z+\frac{\mu}{2}\sum_i(e^{2r_i}-1-2r_i)
                      +\frac{\kappa}{2}\theta^2,\qquad
 p_i=-n\{\mu(e^{2r_i}-1)+\kappa\theta\}.
\]
Thus $W\ge m+z$ and the relaxed pressures vanish.
For an ordinary deformation gradient $F=I+D$, the quadratic reference strain energy is
$\mu|\operatorname{sym}D|^2+\frac{\kappa}{2}(\operatorname{tr}D)^2$.
Korn's inequality removes only infinitesimal rigid motions on a bounded connected body.
The kinetic quadratic form is $(m+z)|\dot u|^2/2$.
The acoustic tensor in a unit direction $\hat k$ is
$\mu I+(\mu+\kappa)\hat k\otimes\hat k$, giving \eqref{eq:bulkspeeds}.
Positivity and the strict gap below unit speed persist by continuity of these symmetric quadratic forms.
The relaxed DEC margin is at least $m$.

\end{proof}

\begin{corollary}[Static finite-thickness reference bodies]
\label{cor:bulkstatic}
Fix a connected material annulus
$\mathcal B=\{X\in\mathbb R^3:R_-<|X|<R_+\}$ with $0<R_-<R_+$,
a smooth fuel profile $z_0(X)\ge0$, and the parameters of Proposition~\ref{prop:bulkcandidate}
with $\kappa+2\mu<m$.
Restore Newton's constant $G$ while retaining $c=1$.
For sufficiently small $G>0$, the material admits an asymptotically flat static
Einstein--elastic configuration near its relaxed state, with a regular vacuum cavity,
traction-free material boundaries, strict material DEC, and subluminal local elastic characteristics.
\end{corollary}
\begin{proof}
Apply Proposition 4.3 and Theorem 5.9 of~\cite{andersson2007} to
$W=\rho/n$. The annulus is bounded, connected, and smooth; the theorem
does not require a connected boundary. At relaxation,
$W=m+z_0\ge m>0$, $D_HW=0$, and for symmetric $E\ne0$,
\[
 D_H^2W(I)[E,E]=\frac\mu2\operatorname{tr}(E^2)
                      +\frac\kappa4(\operatorname{tr}E)^2>0.
\]
These verify conditions (3.20)--(3.22). Projection removes rigid motions
and equilibration restores the field equations. Sobolev embedding for
$p>3$ preserves the cavity, DEC, and causality; zero traction excludes
a surface layer.
\end{proof}

The existence threshold for $G$ is not uniform as the wall becomes
thinner. A nonspherical static cavity is generally curved.

\subsection{Photon emission and recoil}

We couple fuel depletion to a nonnegative photon source.
Write a photon momentum as $p^\mu=E(U^\mu+\omega^\mu)$, where
$U\cdot\omega=0$, $\omega\cdot\omega=1$, and use
$dP=E\,dE\,d\Omega$ on the future mass shell.
Choose a smooth nonnegative $\psi$ supported in a compact positive energy interval with
$4\pi\int E^2\psi(E)\,dE=1$.
Unless stated otherwise, take a smooth rate $\gamma\ge0$ supported away from both material faces
and a smooth material-frame vector $b$ with $|b|\le1$.
Here $\gamma$ is the local fuel-depletion rate, and $b$ controls the
angular asymmetry of emission.
For a finite pulse the rate is also compactly supported in material proper time.
Set
\begin{align}
 {\cal L}_g f=\mathcal C=n\gamma z\,\psi(E)(1+b\cdot\omega),\qquad
 &Uz=-\gamma z,\nonumber\\
 \nabla_\mu T_{\rm mat}^{\mu\nu}=-J^\nu,\quad
 &J^\nu=n\gamma z(U^\nu+b^\nu/3).
 \label{eq:bulkemission}
\end{align}
Here $f$ is the photon distribution and ${\cal L}_g$ is differentiation along the affinely parametrized null geodesic flow on the mass shell.
The photon stress is $T_\gamma^{\mu\nu}=\int p^\mu p^\nu f\,dP$.
Spherical moments give $\nabla_\mu T_\gamma^{\mu\nu}=J^\nu$.
Thus total stress is conserved; the material energy projection gives
fuel depletion and the spatial projection gives recoil. Particle
conservation, $\nabla_\mu(nU^\mu)=0$, is an identity. Fuel stays
nonnegative, the source vanishes at $z=0$, and $f\ge0$ implies photon
DEC. The emitter is a continuum control law.
More generally, replace $1+b\cdot\omega$ by a smooth nonnegative
angular kernel $a_*(\omega)$ and write
\begin{equation}
 \frac1{4\pi}\int a_*\,d\Omega=1,\qquad
 d_*^i=\frac1{4\pi}\int a_*\omega^i\,d\Omega,\qquad
 J=n\gamma z(U+d_*^ie_i),\qquad |d_*|<1.
 \label{eq:generalemission}
\end{equation}
The vector $d_*$ is the mean emission direction weighted by power;
its magnitude is the recoil efficiency. The dipole has $d_*=b/3$.

\subsection{Emission across free faces}

We switch emission on after an initial static interval, so the initial
data satisfy the boundary conditions and their time derivatives.
Conservation of total stress preserves the Einstein constraints in
harmonic coordinates. The remaining issue is regularity at a free face,
which moves with the material and has zero normal traction.
Nearly tangent, or grazing, rays can give $f\notin H^1$ and moments
outside piecewise $H^2$ (appendix~\ref{app:transport}). Here $H^s$
controls square-integrable derivatives through order $s$; piecewise
means separately in the material, cavity, and exterior.
We therefore restrict emission to rays with a uniformly positive
outward crossing velocity.

\label{sec:transverse}

Let $e_r$ be the outward unit material vector proportional to
$g^{\mu\nu}\partial_\nu|X|$.
Choose a smooth nonnegative function $a_0(\mu)$, supported in $\mu>3/4$, with
$\tfrac12\int_{-1}^1a_0(\mu)\,d\mu=1$.
For a material-frame vector $b$, $|b|<1/2$, define
\begin{equation}
 a_b(\omega)=a_0(e_r\cdot\omega)
                   [1+b\cdot(\omega-\beta e_r)],
 \qquad
 \beta=\frac12\int_{-1}^1\mu a_0(\mu)\,d\mu\in(3/4,1).
 \label{eq:outwardkernel}
\end{equation}
Write $\langle\cdot\rangle_0=(4\pi)^{-1}\int(\cdot)a_0\,d\Omega$.
Then $\langle\omega\rangle_0=\beta e_r$ and $\int a_b\,d\Omega=4\pi$.
The bracket is at least $1-(1+\beta)|b|>0$, so the support cone is unchanged.
Replace $(1+b\cdot\omega)$ in \eqref{eq:bulkemission} by $a_b$ and replace
$b/3$ in $J$ by
$d_b=(4\pi)^{-1}\int\omega a_b\,d\Omega$.
The recoil is exactly affine in the control:
\begin{equation}
 d_b=\beta e_r+\mathsf C_{e_r}b,\qquad
 \mathsf C_e=\frac{1-\zeta}{2}(I-e\otimes e)
                  +(\zeta-\beta^2)e\otimes e,
 \quad \zeta=\langle(e_r\cdot\omega)^2\rangle_0.
 \label{eq:conecovariance}
\end{equation}
Smooth $a_0$ gives $\beta^2<\zeta<1$, so the covariance
$\mathsf C_e$ is positive. The mean satisfies $e_r\cdot d_b>3/4$
and $|d_b|<1$. This cone preserves fuel exchange and positivity.
The dipole mass law does not apply to this angular distribution.

The parameter $\eta$ below scales the emission strength.

\begin{theorem}[Local evolution with transverse photon crossings]
\label{thm:transversebody}
Fix a smooth static body of Corollary~\ref{cor:bulkstatic} at sufficiently
small $G>0$. Let $\Gamma(t,X)\ge0$ and $b(t,X)$ be smooth on the closed
material annulus, with $|b|\le b_0<1/2$, and let $\Gamma$ vanish near $t=0$.
Use \eqref{eq:outwardkernel}, zero initial photons, and
$\gamma=\eta\Gamma U^0$, so
$z=z_0\exp[-\eta\int_0^t\Gamma(s,X)\,ds]$.
For sufficiently small $\eta\ge0$ there is a common $T>0$ and a unique
harmonic-coordinate Einstein--elastic--photon evolution on $[0,T]$.
The material faces are free, timelike, and traction-free; material DEC
and causal elasticity persist. Photons remain nonnegative, avoid the
cavity, and cross the outer face transversely.
For any fixed integer $s\ge8$, the material is of order $H^{s+1}$ and
the metric is globally $H^2$ and piecewise $H^{s+1}$ locally in space;
the photons are continuous and piecewise $H^s$ in phase space, with the
corresponding mixed time derivatives.
Emission is allowed at either face.
\end{theorem}
Appendix~\ref{app:transverse} gives the proof.

Emission in a collar of the outer face produces photon crossings during
the evolution. If a pulse is confined to $R_+-d<|X|<R_+$, with $d<T/4$,
and ends before $T/2$, all its photons leave the material by $T$:
the radial-speed bound limits their crossing time to $2d$.

\begin{corollary}[Local momentum control with self-gravity]
\label{cor:gravitatingcontrol}
Choose a smooth radial $z_0>0$ and the spherical static references of
Corollary~\ref{cor:bulkstatic} for $0\le G\le G_0$.
Fix their asymptotically Cartesian harmonic coordinates and zero initial photons.
There are $T,G_0>0$ for which the following holds.
Choose nonnegative $a\in C_c^\infty((0,T/2))$, positive on one interval,
and a smooth nonnegative radial $\chi$ on $\overline{\mathcal B}$, with
\[
 A_0=\int_0^T a(t)\,dt>0,\qquad I_z=\int_{\mathcal B}\chi z_0\,dX>0,
 \qquad K_c=\frac{1-\beta^2}{3}I_zA_0.
\]
Use $\Gamma=a\chi$ and the controls
\begin{align}
 b_*(t;\theta)&=b_c(\cos\phi(t),\sin\phi(t),0)+\theta,
 \quad \phi(t)=\frac{2\pi}{A_0}\int_0^t a(s)\,ds,\nonumber\\
 &0<b_c<1/4,\qquad |\theta|\le r_c=b_c/4.
 \label{eq:conecontrol}
\end{align}
Use the orthonormal material triad
\[
 E_i^\mu=g^{\mu\nu}\partial_\nu X^A(H^{-1/2})_{Ai},\qquad b=b_*^iE_i.
\]
For sufficiently small $G\ge0$ and $\eta>0$, every
$|p_*|\le\eta K_cr_c/2$ is attained by a unique $\theta$ in this family as
\begin{equation}
 P_i(T)=p_{*,i},\qquad
 P_i(t)=\int_{\Omega_t}\sqrt{-\det g}\,T_{\rm mat}^{0}{}_{i}\,d^3x.
 \label{eq:materialmomentumG}
\end{equation}
The total material force has nonparallel values during emission, and all
targets have the same final fuel $z_0e^{-\eta\chi A_0}$.
These are finite-time material momenta in the fixed harmonic frame;
no asymptotic remnant momentum or stationary arrival is asserted.
\end{corollary}
Here $\Omega_t$ is the material region at time $t$. The rotating part of
$b_*$ changes the force direction, while the constant offset $\theta$
sets the final momentum without changing fuel depletion.
Appendix~\ref{app:controlG} gives the proof.

As $\beta\to1$, radial collimation reduces the angular freedom and
hence the control radius.

\section{Fuel and cavity limits at zero coupling}
\label{sec:rigidarrival}

At $G=0$, a nonrotating unstrained body can change its velocity through
local photon recoil. This limit isolates the material force balance:
the photons carry momentum but do not curve spacetime. Keeping the
cavity photon-free imposes an additional restriction.
Here $\tau$ is center proper time and a dot means $d/d\tau$.

\subsection{Recoil at the inner face}

Let $Z$ be the center worldline and $e_i$ an orthonormal spatial frame
transported without local rotation, by Fermi--Walker transport.
The lapse $N$ relates material proper time to center proper time.
Emission quantities at a face are limits taken from within the material.

\begin{proposition}[Stress-free recoil and an empty cavity]
\label{prop:emptycavityrecoil}
In Minkowski spacetime, consider the nonrotating unstrained motion
$Y=Z+e_iX^i$ of material surrounding a bounded smooth cavity
$\mathcal C$ in label space. Let $e_i$ be Fermi--Walker transported,
$A_{\rm acc}$ the center acceleration in this frame, and
$N=1+A_{\rm acc}\cdot X>0$. Assume positive density up to the inner face.
Suppose the material stress is $T_{\rm mat}=\rho U\otimes U$ and its
only force is recoil from a nonnegative photon source, with continuous
energy and momentum moments up to that face. If the cavity remains
photon-free on an open time interval, then $A_{\rm acc}=0$ there.
At $X\in\partial\mathcal C$, let $e$ be the unit normal directed from
cavity into material. Any angular emissivity trace realizing these
moments satisfies
\begin{equation}
 j_-(e):=\int_{\omega\cdot e<0}\mathcal P(\omega)\,d\Omega
 \ge\frac{\rho}{N}(A_{\rm acc}\cdot e)_+.
 \label{eq:inwardemissivity}
\end{equation}
Here $j_-$ is emitted energy per proper volume and time at the face,
not the photon energy flux already crossing it.
\end{proposition}
\begin{proof}
Write the emitted four-momentum per unit proper volume and time as
$J=j_0U+\mathbf j$, where
$j_0=\int\mathcal P(\omega)\,d\Omega$ and
$\mathbf j=\int\mathcal P(\omega)\omega^ie_i\,d\Omega$ for
$\mathcal P\ge0$. The spatial material equation is
\begin{equation}
 \frac{\rho}{N}A_{\rm acc}^ie_i=-\mathbf j.
 \label{eq:emptycavityforce}
\end{equation}
Taking the scalar product with $e$ gives
\[
 \frac{\rho}{N}A_{\rm acc}\cdot e
 =\int_{\omega\cdot e<0}\mathcal P|\omega\cdot e|\,d\Omega
  -\int_{\omega\cdot e\ge0}\mathcal P\,\omega\cdot e\,d\Omega
 \le j_-(e).
\]
Positivity proves \eqref{eq:inwardemissivity}. If the cavity stays empty,
then in a sufficiently thin material collar there is no emission into
directions $\omega\cdot e<-\epsilon$: those rays enter the cavity
transversely in a time tending to zero with the collar thickness.
Thus $\mathbf j\cdot e\ge-\epsilon j_0$ in that collar. Continuity
of the moments, followed by $\epsilon\downarrow0$, gives
$\mathbf j\cdot e\ge0$ at the face. If $A_{\rm acc}\ne0$, maximize
$A_{\rm acc}\cdot X$ over $\overline{\mathcal C}$. At a maximizing
boundary point smoothness gives $e=A_{\rm acc}/|A_{\rm acc}|$.
Equation~\eqref{eq:emptycavityforce} then requires
$\mathbf j\cdot e=-\rho|A_{\rm acc}|/N<0$, a contradiction.
The argument requires neither a spherical nor a convex cavity.

\end{proof}

For the spherical annulus, the forward-pole recoil efficiency during
a nonzero burn is
$\eta_{\rm loc}=\rho|A_{\rm acc}|/(Nj_0)$. Equation~\eqref{eq:inwardemissivity}
gives $j_-/j_0\ge\eta_{\rm loc}$, with $\eta_{\rm loc}=|d_*|$ for a common depletion rate. Near-unit efficiency therefore forces almost
all local emission into the cavity. A saturated dipole has
$j_-/j_0=\frac12\int_{-1}^0(1-q)\,dq=3/4$ there.
These are emissivity fractions, not an estimate of passenger dose.

Under these assumptions, an empty cavity requires stress or another
force during acceleration. The unstrained turns preserve proper
distances between material elements and have zero local rotation.
These are irrotational Born-rigid motions, consistent with
Herglotz--Noether~\cite{hu2010}.

\subsection{Exact unstrained maneuvers}

Distributed emission, including near the material faces, realizes exact
unstrained motion. Photons may enter the cavity. Setting the source
to zero outside the body and integrating along null rays gives a bounded
photon distribution satisfying transport in the weak sense.

\begin{theorem}[An exact unstrained radiating maneuver]
\label{thm:rigidarrival}
Let $\mathcal B$ be the material annulus, $\kappa+2\mu<m$, and $z_0(X)>0$ smooth on its closure.
Choose a smooth compact pulse $\lambda(\tau)\ge0$ and a smooth normalized
angular kernel as in \eqref{eq:generalemission}, with mean $d_*(\tau)$. Put
\begin{equation}
 A=-\lambda d_*,\qquad
 \Lambda(\tau)=\int_{-\infty}^{\tau}\lambda(s)\,ds,
 \qquad \Lambda_f<\min_{\overline{\mathcal B}}\log(1+z_0/m),
 \qquad R_+\sup_\tau|A|<1.
 \label{eq:rigidcontrols}
\end{equation}
Here $\Lambda_f=\Lambda(+\infty)$.
There is an exact globally embedded material motion in Minkowski spacetime that is unstrained throughout and inertial before and after the pulse.
In a Fermi--Walker frame along its center it is
\begin{equation}
 Y(\tau,X)=Z(\tau)+e_i(\tau)X^i,\quad
 \dot Z=v,\quad \dot v=A^ie_i,\quad \dot e_i=A_i v,
 \qquad
 z=(m+z_0)e^{-\Lambda}-m.
 \label{eq:rigidmotion}
\end{equation}
The material obeys strict DEC, causal elasticity, exact fuel exchange,
and zero traction. Its final rest-frame state is relaxed.
If $\mathcal M=\int_{\mathcal B}(m+z)\,dX$ is the material rest energy and $L=\int|A|\,d\tau$ is the center's velocity-space path length, then
\begin{equation}
 \mathcal M_f/\mathcal M_i=e^{-\Lambda_f}<e^{-L}\quad(L>0).
 \label{eq:rigidbudget}
\end{equation}
For dipole emission the stronger bound is $e^{-3L}$, attained when
$|b_*|=1$ during emission. Noncollinear recoil directions give turning
subject to \eqref{eq:rigidcontrols}.
\end{theorem}
\begin{proof}
The pullback of Minkowski space is
\[
 ds^2=-N^2d\tau^2+\delta_{ij}dX^idX^j,\qquad
 N=1+A\cdot X>0.
\]
Thus $U=N^{-1}\partial_\tau$, $H=I$, $n=1$, $\nabla\cdot U=0$, and the proper acceleration is $A^ie_i/N$.
Along every future timelike fixed-$X$ generator, $v(\tau_1)\cdot Y$
is strictly decreasing. The entire $\tau_1$ section lies in one level
hyperplane, so no other section meets it. This proves global injectivity;
the nonsingular differential and inertial ends give a regular worldtube.

At $H=I$ the elastic stress vanishes and
$T_{\rm mat}=(m+z)U\otimes U$.
The fuel bound gives $z>0$. Choose
$\gamma=\lambda(m+z)/(Nz)$.
Then $Uz=-\gamma z$, and direct differentiation gives
\[
 \nabla_\mu T_{\rm mat}^{\mu\nu}
 =-\frac{\lambda(m+z)}N U^\nu+\frac{m+z}N A^ie_i^\nu
 =-\gamma z(U^\nu+d_*^ie_i^\nu).
\]
This is \eqref{eq:generalemission}.
Tangency of $U$ gives zero traction and no surface divergence. Null-line
integration gives a unique bounded nonnegative photon distribution with
zero incoming data and divergence opposite to the material.
Proposition~\ref{prop:bulkcandidate} gives DEC and causality.

Outside the pulse the body is relaxed and inertial. Integration gives
$\mathcal M_f/\mathcal M_i=e^{-\Lambda_f}$ and
$L=\int\lambda|d_*|\,d\tau<\Lambda_f$ for a nontrivial burn.
For a dipole, $3|A|=\lambda|b_*|\le\lambda$.
Slowing the path enforces the lapse bound without changing its fuel cost.
\end{proof}

\begin{lemma}[Local realization of a prescribed maneuver]
\label{lem:localprogram}
For the controls of Theorem~\ref{thm:rigidarrival}, put
$N_0(s,X)=1-\lambda(s)d_*(s)\cdot X$. A local material realization
uses a scalar phase $s$, initialized on the inertial section
$\tau=\tau_0$ before the pulse by $s=\tau_0$, with
\begin{equation}
 Us=N_0(s,X)^{-1},\qquad
 \gamma=\frac{\lambda(s)(m+z)}{N_0(s,X)z}.
 \label{eq:localprogram}
\end{equation}
Evaluate angular components in the local material triad $E_i$ following
\eqref{eq:conecontrol}. On the exact motion, $s=\tau$ and this law
reproduces the prescribed emission. Its additional characteristic is
the timelike material flow.
\end{lemma}
\begin{proof}
The lapse and fuel bounds keep $N_0,z>0$. Since $UX=0$, the phase
equation is a smooth ordinary differential equation on each worldline.
On \eqref{eq:rigidmotion},
$E_i=e_i$ and $U\tau=1/N_0(\tau,X)$, so uniqueness gives $s=\tau$.
All rates use local state and stored functions. The first-order material
system gains a transport block, with no additional derivatives in the
mechanical source. Phase signals propagate with $U$.
\end{proof}
The schedule and synchronization are initial data for the continuum
source. Subsequent commands require causal communication.

\subsection{Sharp fuel bounds}

The minimum $F_0$ below expresses the requirement that every material
element retain fuel throughout the maneuver.

\begin{corollary}[Optimal fuel threshold for unstrained steering]
\label{cor:optimalfuel}
For the material and positive initial fuel of Theorem~\ref{thm:rigidarrival},
a finite nonconstant velocity-space polygon or smooth regular curve
on a compact interval, of length $L$ and with freely chosen timing, can be
executed with smooth nonnegative angular emission and positive final
fuel if and only if
\begin{equation}
 L< F_0:=\min_{\overline{\mathcal B}}\log(1+z_0/m).
 \label{eq:optimalfuel}
\end{equation}
The supremum of the final rest-energy ratio is $e^{-L}$; smooth
angular densities approach but do not attain it. Dipole emission
instead requires $3L<F_0$.
\end{corollary}
\begin{proof}
For $b\in\mathbb R^3$, $|b|<1$, the spherical Poisson kernel is
\begin{equation}
 \mathcal A_b(\omega)=\frac{1-|b|^2}{(1-2b\cdot\omega+|b|^2)^{3/2}},\qquad
 \frac1{4\pi}\int \mathcal A_b\,d\Omega=1,\qquad
 \frac1{4\pi}\int \mathcal A_b\omega\,d\Omega=b.
 \label{eq:poissonemission}
\end{equation}
Aligning $b=ce_3$ and substituting $v=1+c^2-2cq$ gives both
moments by elementary integration; transverse moments vanish by symmetry.

Choose $L/F_0<c<1$ and parameterize the path by arc length $\ell$.
Use $\lambda=\dot\ell/c$ and mean $b=-c\hat A$ in the transported
frame, with smooth timing flat at the ends. For polygons, stop at each
vertex before changing the axis. Then $\Lambda_f=L/c<F_0$;
slowing the path enforces $R_+|A|<1$.
Conversely, at each material label the stress-free equations give
$\partial_\tau\rho=-Nj_0$ and $\rho|A|=N|\mathbf j|\le Nj_0$.
Hence $\rho_f\le(m+z_0)e^{-L}$; positive final fuel requires
$\rho_f>m$, and therefore $L<F_0$.
Letting $c\uparrow1$ gives the unattained ideal limit~\cite{fuzfa2019}.
The dipole uses $c=1/3$, giving its threshold and saturated construction.
\end{proof}

A finite upper bound on angular power prevents perfect collimation.
Here $B$ bounds power per solid angle relative to its angular average.
\begin{corollary}[Bounded angular concentration]
\label{cor:boundedemission}
Impose $0\le a_*\le B$ in \eqref{eq:generalemission}, with fixed $B>1$,
and put $c_B=1-B^{-1}$. For the paths and material of
Corollary~\ref{cor:optimalfuel}, smooth emission with positive final
fuel exists if and only if
\begin{equation}
 L<c_BF_0,\qquad
 \sup\frac{\mathcal M_f}{\mathcal M_i}=e^{-L/c_B}.
 \label{eq:boundedfuel}
\end{equation}
For a feasible path with $L>0$, the supremum is not attained.
\end{corollary}
\begin{proof}
Write $\langle f\rangle=(4\pi)^{-1}\int f\,d\Omega$ and align a unit
axis $e$ with $\langle a_*\omega\rangle$ when this is nonzero.
Set $q=e\cdot\omega$, $q_B=1-2/B$, and
$a_B=B\,\mathbf1_{\{q>q_B\}}$. Direct integration with
$\langle f(q)\rangle=\frac12\int_{-1}^1f(q)\,dq$ gives
$\langle a_B\rangle=1$ and $\langle qa_B\rangle=c_B$.
Since $(q-q_B)(a_B-a_*)\ge0$ pointwise,
\begin{equation}
 c_B-|d_*|=\langle(q-q_B)(a_B-a_*)\rangle\ge0.
 \label{eq:capdeficit}
\end{equation}
Equality requires the discontinuous cap $a_B$. Smooth caps
$a_*(q)=BH((q-q_B)/\epsilon)$ approach it, where $H$ increases
smoothly from $0$ to $1$ on $[-1,1]$, $H(-x)=1-H(x)$, and
$0<\epsilon<\min(1-q_B,1+q_B)$. This symmetry preserves normalization.
Mixing with isotropic emission realizes every mean $0\le c<c_B$.

The pointwise material balance gives
$\rho|A|=N|\mathbf j|\le c_BNj_0=-c_B\partial_\tau\rho$.
Thus positive final fuel requires $L<c_BF_0$, and integration over
material labels bounds the final rest energy by $e^{-L/c_B}$.
The construction with $L/F_0<c<c_B$ proves sufficiency;
$c\uparrow c_B$ proves sharpness.
\end{proof}

At $B=1$ the emission is isotropic and recoil vanishes.

\subsection{Concentration and inward emission}

We also bound the fraction $f_-$ of emitted power entering the cavity
at its forward inner pole. This restriction and the angular cap $B$
jointly limit recoil.

\begin{proposition}[Joint concentration and inward-emission bound]
\label{prop:inwardfuel}
At the forward inner pole of a nonzero unstrained burn, let
$q=-\omega\cdot\widehat A$, $\langle\cdot\rangle=(4\pi)^{-1}\int\cdot\,d\Omega$,
and $f_-=\langle a_*\mathbf1_{\{q>0\}}\rangle$.
For $\langle a_*\rangle=1$ and $0\le a_*\le B$, $B>1$,
\begin{equation}
 \eta_{\rm loc}=\langle qa_*\rangle
 \le \eta_B(f_-):=f_--\frac{f_-^2+(1-f_-)^2}{B}.
 \label{eq:inwardconcentration}
\end{equation}
Impose $f_-\le\varepsilon$, $0\le\varepsilon\le1$, at this pole
during every nonzero burn, and set
$\bar\varepsilon=\min\{\varepsilon,B/2\}$,
$c_{B,\varepsilon}=\eta_B(\bar\varepsilon)$.
If $c_{B,\varepsilon}\le0$, no such smooth burn is possible.
If $c_{B,\varepsilon}>0$, the maneuvers of
Theorem~\ref{thm:rigidarrival}, with a common depletion rate and angular
kernel at every material label, execute a path of
Corollary~\ref{cor:optimalfuel} of length $L$
with positive final fuel exactly when
\begin{equation}
 L<c_{B,\varepsilon}F_0,\qquad
 \sup\frac{\mathcal M_f}{\mathcal M_i}
       =e^{-L/c_{B,\varepsilon}}.
 \label{eq:inwardfuel}
\end{equation}
Timing is free; smooth kernels approach but do not attain the supremum.
\end{proposition}
\begin{proof}
Each hemisphere has capacity $B/2$, so
$\max\{0,1-B/2\}\le f_-\le\min\{1,B/2\}$.
Fix $f=f_-$ and set $u=1-2f/B$, $l=-2(1-f)/B$.
The comparison kernel
$a^\sharp=B\mathbf1_{\{u<q<1\}}+B\mathbf1_{\{l<q<0\}}$
has the same mass in each hemisphere. On $q>0$,
$(q-u)(a^\sharp-a_*)\ge0$; on $q<0$ the same holds with $u$
replaced by $l$. Integration cancels the threshold terms, yielding
\eqref{eq:inwardconcentration}, since
\[
 \langle qa^\sharp\rangle
 =\frac B4(1-u^2-l^2)=\eta_B(f).
\]
Equality requires $a_*=a^\sharp$ almost everywhere, which excludes
a smooth kernel. Also, $\eta_B'(f)=1+(2-4f)/B\ge0$
on the feasible interval.

For smooth sharpness take $f$ strictly inside that interval and the step $H$ used in
Corollary~\ref{cor:boundedemission}. For sufficiently small $\delta>0$,
\[
 a_\delta(q)=B\left[
 H\left(\frac{q-u}{\delta}\right)
 +H\left(\frac{q-l+\delta}{\delta}\right)
 -H\left(\frac{q+\delta}{\delta}\right)\right].
\]
The first term is supported in $q>0$; the difference is supported
in $q<0$. Their masses are exactly $f$ and $1-f$ by the symmetry
of $H$. The supports are disjoint, $0\le a_\delta\le B$, and
$\langle qa_\delta\rangle\to\eta_B(f)$.
Positivity of $c_{B,\varepsilon}$ places $\bar\varepsilon$ above
the lower feasible endpoint: for $1<B<2$ the value there is
$-(B-1)^2/B$, and for $B\ge2$ it is $-1/B$.
Taking $f\le\varepsilon$ to $\bar\varepsilon$, then $\delta\downarrow0$,
gives positive means approaching $c_{B,\varepsilon}$.

For common depletion, $\rho=(m+z_0)e^{-\Lambda}$ and
$|A|=\lambda\eta_{\rm loc}$, so
$L\le c_{B,\varepsilon}\Lambda_f<c_{B,\varepsilon}F_0$.
Conversely choose a smooth comparison kernel with mean
$L/F_0<c<c_{B,\varepsilon}$ and set $\lambda=\dot\ell/c$ along
the path. Rotate its axis during zero-rate coasts at polygon vertices.
The proof of Corollary~\ref{cor:optimalfuel} then applies. Taking
$c\uparrow c_{B,\varepsilon}$ proves the mass supremum.
\end{proof}
For $B\ge2$, $\bar\varepsilon=\varepsilon$. With no inward restriction,
$\varepsilon=1$ recovers $c_{B,1}=1-B^{-1}$.
The lower root of $\eta_B(f)=0$ gives the explicit threshold
\begin{equation}
 \varepsilon>\varepsilon_{\rm crit}(B)
 :=\frac{B+2-\sqrt{(B+2)^2-8}}4,
 \qquad B>1.
 \label{eq:inwardthreshold}
\end{equation}
This condition is necessary and sufficient for a positive recoil
efficiency under the two angular restrictions. Indeed, $\eta_B$ is
increasing on its feasible interval, is negative at its lower endpoint,
and reaches $1-B^{-1}>0$ at its upper endpoint. Thus a positive inward
allowance alone is insufficient at finite $B$. For example,
$\varepsilon_{\rm crit}(2)=1-1/\sqrt2$; as $B\to\infty$,
$\varepsilon_{\rm crit}(B)=B^{-1}+O(B^{-2})$.

\subsection{Distance, duration, and fuel}

The fuel bounds also determine which positions can be reached at a
prescribed time. We impose a bound on the center's proper acceleration,
as in the relativistic rocket problem~\cite{henriques2012}.

\begin{proposition}[Distance, duration, and fuel]
\label{prop:timedarrival}
Consider the unstrained $G=0$ body of Theorem~\ref{thm:rigidarrival},
initially and finally at rest in one inertial frame. Let $D$ be the
magnitude of its center displacement over coordinate time $T>0$, and
impose $|A|\le a_{\max}<1/R_+$, with $a_{\max}>0$. Define
\begin{equation}
 r_T=\operatorname{arsinh}(a_{\max}T/2),\qquad
 \mathcal D_T(r)=T\tanh r-\frac{2}{a_{\max}}(1-\operatorname{sech}r).
 \label{eq:distancetime}
\end{equation}
Every maneuver of velocity-space length $L$ satisfies
\begin{equation}
 D\le\mathcal D_T\bigl(\min\{L/2,r_T\}\bigr)
       \le T\tanh(L/2).
 \label{eq:distancebound}
\end{equation}
For $D>0$, the first inequality is strict when the acceleration is
smooth and supported in $(0,T)$.

Let $c_*=1$, $1/3$, or $1-B^{-1}$ for unrestricted emission, dipole
emission, or angular concentration bounded by $B>1$, respectively.
In the class with common
depletion and angular kernels of Proposition~\ref{prop:inwardfuel},
instead take $c_*=c_{B,\varepsilon}>0$.
With $F_0$ from \eqref{eq:optimalfuel}, a prescribed displacement vector
of magnitude $D>0$ is attainable at time $T$ with a smooth pulse
supported in $(0,T)$, rest at both ends, and positive final fuel if and only if
\begin{equation}
 D<\mathcal D_T\bigl(\min\{c_*F_0/2,r_T\}\bigr).
 \label{eq:timedfuelthreshold}
\end{equation}
For a feasible target, let $r_D\in(0,r_T)$ solve $\mathcal D_T(r_D)=D$.
This solution is unique, and
\begin{equation}
 \sup\frac{\mathcal M_f}{\mathcal M_i}=e^{-2r_D/c_*}.
 \label{eq:timedmass}
\end{equation}
The supremum is not attained, even for dipole emission.
\end{proposition}
For $0\le r\le r_T$, $\mathcal D_T(r)$ is the limiting displacement
of an accelerate--coast--brake trajectory with peak rapidity $r$.
At $r=r_T$ the coast disappears. Smooth pulses approach this limit.
\begin{proof}
Write the center four-velocity as $(\cosh\chi,\sinh\chi\,\hat n)$.
The lengths of the velocity-space path before and after any point are
both at least $\chi$, so $\chi\le L/2$.
Also $|d\chi/d\tau|\le|A|$. Hence $p=\sinh\chi$ satisfies
$|dp/dt|\le a_{\max}$ almost everywhere, since $dt/d\tau=\cosh\chi$.
The endpoint rest conditions give
\begin{equation}
 p(t)\le\min\{a_{\max}t,a_{\max}(T-t),\sinh(L/2)\}.
 \label{eq:speedenvelope}
\end{equation}
Displacement is at most the integrated speed $p/\sqrt{1+p^2}$.
At $r=\min\{L/2,r_T\}$ the envelope has two ramps of duration
$\sinh r/a_{\max}$, contributing $2(\cosh r-1)/a_{\max}$,
and a coast contributing $(T-2\sinh r/a_{\max})\tanh r$.
Their sum is $\mathcal D_T(r)$, proving \eqref{eq:distancebound}.
Equality requires collinear motion and saturation of the envelope;
a smooth acceleration supported in $(0,T)$ cannot saturate its initial ramp.

For $0\le r<r_T$,
$\mathcal D_T'(r)=\operatorname{sech}^2r
(T-2\sinh r/a_{\max})>0$.
The preceding fuel bounds $L<c_*F_0$ and
$\mathcal M_f/\mathcal M_i\le e^{-L/c_*}$ therefore give
necessity and the strict bound $\mathcal M_f/\mathcal M_i<e^{-2r_D/c_*}$.
For sufficiency choose $r_D<r<\min\{c_*F_0/2,r_T\}$ and an allowed
smooth kernel with mean magnitude $c$ such that $2r<cF_0$.
Take $c=1/3$ for the dipole; in the other classes take $c<c_*$
sufficiently close to its ceiling.

For small $\delta>0$, convolve
$[\min\{a_{\max}(t-2\delta),a_{\max}(T-2\delta-t),\sinh r\}]_+$
with a nonnegative even smooth unit mollifier supported in
$(-\delta,\delta)$, where $[x]_+=\max\{x,0\}$.
The resulting momentum is smooth, supported in $(0,T)$, and has
one rise, a nonempty coast, and one fall, with $|p'|\le a_{\max}$.
Its displacement tends to $\mathcal D_T(r)>D$ as $\delta\downarrow0$.
Scaling its amplitude by a constant in $(0,1)$ gives exactly $D$
by continuity. Direct this motion along the target vector.
Its signed proper acceleration is $p'(t)$ and
$L=2\operatorname{arsinh}(\max p)\le2r$.
The separated smooth ramps allow $\lambda=|A|/c$ with a smooth
emission-axis reversal during the coast. Theorem~\ref{thm:rigidarrival}
then gives the exact body with $\Lambda_f=L/c<F_0$ and $N>0$.
Letting $r\downarrow r_D$, $\delta\downarrow0$, and $c\uparrow c_*$
when needed proves \eqref{eq:timedmass}.
\end{proof}

\section{Discussion and conclusions}
\label{sec:comparison}\label{sec:limits}\label{sec:conclusion}

Exact timelike junctions give compact flat cavities with strict surface
DEC and nonnegative exterior radiation. Static collars permit successive
changes of direction at mass ratio $e^{-3L}$. The supported center has
its own sharp bound because its velocity is defined by the interior
embedding. Neither bound describes geodesic transport of a payload.

Material dynamics impose further conditions. The self-similar shells
require anisotropy, and compression produces growing normal modes in
the frozen surface equations on a prescribed geometry. Finite thickness
provides elastic stresses and local self-gravitating recoil control
with outward emission. It also changes the geometry: the bulk model
has a dynamical exterior and a generally curved vacuum cavity.

At $G=0$, exact unstrained turns make the fuel cost explicit. Smooth
angular emission approaches the ideal photon-rocket efficiency;
bounded concentration and limited inward emission reduce it by sharp
amounts. At finite angular concentration, nonzero recoil requires the
strictly positive inward-emission threshold \eqref{eq:inwardthreshold}.
With bounded center acceleration, these efficiencies also give the
sharp distance and fuel limits at fixed travel time in
Proposition~\ref{prop:timedarrival}.
A smooth bounded cavity of any shape is incompatible with nonzero
stress-free recoil while it remains photon-free, under the stated
inner-face assumptions. A photon-free cavity
therefore requires stress transmission or an additional force.

The remaining problem is a self-gravitating material cavity that
completes a turn and settles. Local recoil control does not establish
that endpoint. It requires global nonspherical Einstein--matter
estimates and a constitutive mechanism that dissipates the resulting
material motion.

\appendix
\addtocontents{toc}{\protect\setcounter{tocdepth}{1}}

\section{Self-similar matching}
\label{app:homothetic}

The first junction condition reduces to flatness of a two-dimensional
Lorentzian metric. We integrate its angular equation, construct the
interior embedding, and evaluate the surface stress.

For $r=R(u,q)$, set $w=1-q^2$ and let $\rho=R\sqrt w$ be the
cylindrical radius.
The azimuthal matching fixes $\rho$ in the Minkowski embedding
$(T,\rho\cos\varphi,\rho\sin\varphi,z)$; the remaining condition is
\begin{equation}
 b:=h_{(u,q)}-d\rho^2=-dT^2+dz^2,
 \label{eq:quotientmatch}
\end{equation}
where
\begin{align}
 b_{uu}&=-1+2m/R-2aRq+a^2R^2w-2R_u-wR_u^2,\nonumber\\
 b_{uq}&=aR^2-(1+wR_u)R_q+qRR_u,\qquad
 b_{qq}=R^2+2qRR_q-wR_q^2.
 \label{eq:quotientmetric}
\end{align}
This Lorentzian quotient, regular at the poles, admits local matching
exactly when it is flat. With $dt=du/\mathcal R$, $\mathcal R=R_0e^{-\nu t}$,
and $p=\chi'$,
\begin{align}
 b&=\mathcal R^2(E\,dt^2+2F\,dt\,dq+G\,dq^2),
 \label{eq:homoquotient}\\
 E&=-1+x/\chi-2\lambda\chi q+(\lambda^2-\nu^2)\chi^2w+2\nu\chi,\nonumber\\
 F&=\lambda\chi^2-p+\nu w\chi p-\nu q\chi^2,\qquad
 G=\chi^2+2q\chi p-wp^2,\qquad D=F^2-EG.
 \label{eq:homoEFG}
\end{align}
For $E<0$, $D>0$, flatness is $[(E'+2\nu F)/\sqrt D]'=0$.
Its signed first integral is
\begin{align}
 N:=E'+2\nu F
 &=\left(-\frac{x}{\chi^2}-2\lambda q+2\lambda^2\chi w\right)p
       -2\lambda\chi(1-\nu\chi+\lambda q\chi)=2\kappa\sqrt D,
 \label{eq:homofirst}\\
 C&=\sqrt{1-x-2\nu+\nu^2},\qquad
 \kappa=-\frac{\lambda(1-\nu)}C,\qquad
 \mathcal H=N-2\kappa\sqrt D.
 \label{eq:homokappa}
\end{align}
The choice of $\kappa$ enforces $\chi(0)=1$, $p(0)=0$.

\paragraph{Continuation and embedding.}
Choose $d>0$ with $1-d>x$. At $\lambda=\nu=p=0$, uniformly for
$|q|\le1$, $|\chi-1|\le d$, one has
$E=-(1-x/\chi)<0$, $D=\chi^2(1-x/\chi)>0$, and
$\mathcal H_p=-x/\chi^2<0$.
The analytic implicit-function theorem gives $p=V$ with
$|V|\le M(|\lambda|+|\nu|)<d/2$. Integration from $\chi(0)=1$
then reaches both poles inside the neighborhood. Analytic ODE
continuation~\cite{teschl2012} gives derivative bounds and
\begin{equation}
 \chi(q)=1+\frac{\lambda^2q^2}{1-x}
           +O((|\lambda|+|\nu|)^3).
 \label{eq:homoprofile}
\end{equation}
Put $h=\int_0^q F/E\,dq'$, $\zeta=\int_0^q\sqrt D/(-E)\,dq'$,
$\eta=t+h$, and $c_0=R_0\sqrt{-E(0)}$. The first integral gives
\begin{equation}
 b=c_0^2e^{-2\nu\eta-2\kappa\zeta}(-d\eta^2+d\zeta^2).
 \label{eq:homoconformal}
\end{equation}
For $I_c(v)=\int_0^v e^{-cs}\,ds$, including $c=0$, define
\begin{equation}
 U=c_0I_{\nu+\kappa}(\eta+\zeta),\quad
 V=c_0I_{\nu-\kappa}(\eta-\zeta),\quad
 T=(U+V)/2,\quad z=(U-V)/2.
 \label{eq:homodevelop}
\end{equation}
Then $b=-dU\,dV=-dT^2+dz^2$. Monotonicity of $I_c$ and
$\zeta'>0$ makes the development injective, with
\begin{equation}
 T_t>0,\qquad \det\frac{\partial(T,z)}{\partial(t,q)}
 =\mathcal R^2\sqrt D>0,\qquad
 \left.\frac{\partial z}{\partial q}\right|_T
 =\frac{\mathcal R^2\sqrt D}{T_t}>0.
 \label{eq:homojacobian}
\end{equation}
In Cartesian sphere coordinates the embedding is
\begin{equation}
 (u,n)\longmapsto
 (T(u,n_3),\mathcal R\chi(n_3)n_1,\mathcal R\chi(n_3)n_2,z(u,n_3)),
 \qquad n\in S^2.
 \label{eq:homosphere}
\end{equation}
Its rank is three, including at the poles. Complete constant-$T$
sections on an extended slab have increasing height and positive radius,
hence bound balls. Timelikeness persists from the static sphere.
The temporal function $u(T,z)$, since $G>0$, extends through the
exterior by \eqref{eq:endpoint_normalnorm}--\eqref{eq:endpoint_temporal}.

\paragraph{Stress and acceleration.}
For $\ell=r-\mathcal R\chi$ and $H=g_+^{-1}(d\ell,d\ell)>0$, use
$N^+=d\ell/\sqrt H$ and $K^+_{ij}=e_i^ae_j^b\nabla_aN^+_b$.
On the interior, for $i,j\in\{u,q\}$,
\begin{equation}
 N^-_\rho=\frac{\sqrt{wD}}{\chi\sqrt H},\qquad
 K^-_{ij}=-N^-_\rho(\operatorname{Hess}_b\rho)_{ij},\qquad
 K^-_{\varphi\varphi}=N^-_\rho\rho.
 \label{eq:homocurvatures}
\end{equation}
Equation~\eqref{eq:homosphere} supplies the regular pole limits.
Metric scaling gives $S_{\hat a\hat b}=\mathcal R^{-1}\mathcal S_{\hat a\hat b}$
with smooth, time-independent $\mathcal S$. The normalized Type~I DEC
margin of \eqref{eq:decfactor} persists uniformly over the sphere.
Gauss--Codazzi supplies the flux balances.
The equatorial center has $d\tau_\Gamma/du=\sqrt{C^2-\lambda^2}$
and rapidity $-\kappa t$, hence
\begin{equation}
 |a_{\rm int}|=\frac{|\kappa|}{\mathcal R\sqrt{C^2-\lambda^2}}
 =\frac{|\lambda|(1-\nu)}{\mathcal R C\sqrt{C^2-\lambda^2}}.
 \label{eq:homoaccel}
\end{equation}
Smallness ensures $\nu<1$, $C^2>\lambda^2$, and nonzero acceleration
when $\lambda\ne0$.

\paragraph{Equatorial anisotropy.}
At $(q,\chi,p)=(0,1,0)$,
$\mathcal H_p=-x(1+2\lambda^2/C^2)$ and
$\mathcal H_q=2x\lambda^2/C^2$, so
\begin{equation}
 \chi''(0)=\frac{2\lambda^2}{C^2+2\lambda^2}.
 \label{eq:equatorialcurvature}
\end{equation}
Let $J_{ab}=\mathcal R[K_{\hat a\hat b}]$ in the tetrad following
constant angular labels, with meridional leg toward increasing $q$.
Israel's equation gives
$8\pi\mathcal R(\sigma,p_1,p_2,j)
=(-J_{11}-J_{22},J_{22}-J_{00},J_{11}-J_{00},-J_{01})$, and thus
\begin{equation}
 (8\pi\mathcal R)^2\mathcal I
 =-(J_{11}+J_{00})(J_{22}-J_{11})-J_{01}^2.
 \label{eq:anisotropycurvature}
\end{equation}
At $\nu=q=0$, with $s=\sqrt{1-x}$, substitution gives
\begin{align}
 J_{00}&=\frac{s^2-1}{2s}+\frac{3(s^2-1)}{2s^3}\lambda^2+O(\lambda^4),\nonumber\\
 J_{01}&=\frac{3(s^2-1)}{2s^2}\lambda+O(\lambda^3),\qquad J_{22}=s-1,\nonumber\\
 J_{11}&=s-1+\frac{3(s^2+2s-3)}{2s^3}\lambda^2+O(\lambda^4).
 \label{eq:equatorialjumps}
\end{align}
The coefficient in \eqref{eq:anisotropycurvature} is
$3(1-s)^2[(3s+1)(s+3)-3(s+1)^2]/(4s^4)=3(1-s)^2/s^3>0$.
Analyticity, spherical symmetry at $\lambda=0$, and reflection
$\lambda\mapsto-\lambda$ give a factor $\lambda^2$ in $\mathcal I$.
This proves \eqref{eq:anisotropyexpansion} on $\nu\ge3|\lambda|$.

\section{Smooth matching through static endpoints}
\label{app:endpointburn}

We first solve the matching equation analytically in angle on small
time neighborhoods. Pulses whose derivatives vanish at the endpoints
then allow these local solutions to join static intervals smoothly.
The final steps establish a global cavity and concatenate changes of direction.

\paragraph{The exact equation and its type.}
For the general quotient \eqref{eq:quotientmetric}, write
$E=b_{uu}$, $F=b_{uq}$, $G=b_{qq}$, and $\Delta=EG-F^2$.
Flatness is equivalent, when $\Delta\ne0$, to
\begin{align}
 \mathcal P&:=4\Delta\,\operatorname{Riem}_{uquq}[b]=0,\nonumber\\
 \mathcal P&=\Delta(4F_{uq}-2E_{qq}-2G_{uu})
   +B_0 J B_0^{\mathsf T}-A_0 J C_0^{\mathsf T},\label{eq:endpoint_curvature}\\
 J&=\begin{pmatrix}G&-F\\-F&E\end{pmatrix},\qquad
 A_0=(E_u,2F_u-E_q),\nonumber\\
 B_0&=(E_q,G_u),\qquad C_0=(2F_q-G_u,G_q).\nonumber
\end{align}
All third derivatives of $R$ cancel; $\mathcal P$ is affine in $R_{qq}$.

\begin{lemma}[Principal symbol and round-graph obstruction]
\label{lem:endpoint_symbol}
Set $A=\mathcal P_{R_{uu}}$, $B=\mathcal P_{R_{uq}}/2$ and
$C=\mathcal P_{R_{qq}}$. Then
\begin{equation}
 B^2-AC=8mR\mathcal P+\frac{48m^2R_q^2}{R^2}\Delta.
 \label{eq:endpoint_symbol}
\end{equation}
On a solution with $m,R>0$ and $\Delta<0$, the graph equation is elliptic
where $R_q\ne0$ and has zero discriminant where $R_q=0$.
For a round graph $R=R(u)$ with $m,R>0$, $\mathcal P=8m\propacc^2R^3$; hence exact matching
requires $\propacc=0$.
\end{lemma}

\begin{proof}
Substitution and differentiation in the second derivatives give
$A=-8m(RR_{qq}-2R_q^2)$ and
$\mathcal P_{R_{uu}R_{qq}}=-8mR$,
$\mathcal P_{R_{uq}R_{uq}}=16mR$.
Expansion gives \eqref{eq:endpoint_symbol}; setting $R_q=R_{uq}=R_{qq}=0$
gives the last identity.
\end{proof}

At $R=R_0$, $\propacc=0$, the quotient is flat for varying $m$, with
\begin{equation}
 D_R\operatorname{Riem}_{uquq}[Y]
 =\left(\frac m{R_0^2}-\frac{\dot m}{R_0(1-2m/R_0)}\right)Y_{qq}.
 \label{eq:endpoint_pivot}
\end{equation}
The coefficient is positive for $0<2m/R_0<1$ and $\dot m\le0$.
The static symbol has rank one; away from $R_q=0$ the equation is elliptic.

\paragraph{An analytic estimate uniform in the time radius.}
Rescale lengths so that $R_0=1$. Fix a compact interval
$I\Subset(0,1/2)$ of reference masses.
The norm $\|\cdot\|_\rho$ is the supremum on the complex time disk
of radius $\rho$ specified below.

\begin{lemma}[Analytic matching over the sphere]
\label{lem:endpoint_analytic}
There are $c,C>0$, uniform for $m_c\in I$ and $0<\rho\le1$, with the
following property. If $m,\propacc$ are holomorphic near
$|u-u_c|\le\rho$, real on its real diameter, and
\begin{equation}
 e:=\|m-m_c\|_\rho+\|\propacc\|_\rho\le c\rho^8,
 \label{eq:endpoint_smallness}
\end{equation}
then $\mathcal P=0$ has an analytic solution with
$R(u,0)=1$, $R_q(u,0)=0$ on
$|u-u_c|<5\rho/16$, $|q|<9/8$, satisfying
\begin{equation}
 \|(R-1,R_u,R_q,R_{uu},R_{uq})\|_\infty\le C\sqrt e.
 \label{eq:endpoint_jetbound}
\end{equation}
It is unique among nearby analytic solutions with these equatorial data.
On fixed smaller product domains,
$|\partial_u^k\partial_q^\ell(R-1)|\le C_{k\ell}\rho^{-k}\sqrt e$.
\end{lemma}

\begin{proof}
Uniformly for $m_c\in I$, $|q|\le8$, the static values of $\Delta$
and $\mathcal P_{R_{qq}}$ are nonzero. Dividing by the latter gives
\begin{equation}
 \begin{gathered}
 R_{qq}=H(q,J;d),\qquad d=(m-m_c,a,\dot m,\dot a),\\
 J=(Y,U,V,A_2,B_2)=(R-1,R_u,R_q,R_{uu},R_{uq}).
 \end{gathered}
 \label{eq:endpoint_normal}
\end{equation}
On a smaller fixed polydisk derivatives are bounded uniformly by $M$, and
$H(q,0;0)=D_JH(q,0;0)=0$.

Treat the derivatives collected in $J$ as independent unknowns and
solve the following first-order system with zero data at $q=0$:
\begin{align}
 Y_q&=V,\qquad U_q=B_2,\qquad V_q=H,\qquad
 (A_2)_q=(B_2)_u,\nonumber\\
 (B_2)_q&=H_YU+H_UA_2+H_VB_2\nonumber\\
 &\quad+H_{A_2}(A_2)_u+H_{B_2}(B_2)_u+H_d d_u.
 \label{eq:endpoint_jets}
\end{align}
Use bounded holomorphic functions on disks of radius
$r(s)=\rho(1+s)/4$, $0<s<1$, with component maximum norm.
Cauchy's estimate is
\[
 \|\partial_u W\|_{s'}\le\frac4{\rho(s-s')}\|W\|_s,
 \qquad |d|\le C e/\rho,\quad |d_u|\le C e/\rho^2.
\]
The static linear part is $LJ+NJ_u$, where
$L_{13}=L_{25}=N_{45}=1$ and the remaining entries vanish.
For
\[
 J=\delta S_\eta W,\qquad
 S_\eta=\operatorname{diag}(1,1,\eta,1,\eta),
 \qquad 0<\eta,\delta\le1,
\]
one has $S_\eta^{-1}LS_\eta=\eta L$ and
$S_\eta^{-1}NS_\eta=\eta N$.

The scaled field $\mathcal F_q(W)$ on $\|W\|_s<1$ satisfies the
analytic-scale bounds
\begin{align}
 \|D\mathcal F_q(W)Z\|_{s'}
 &\le\frac{\mathcal C}{s-s'}\|Z\|_s,\nonumber\\
 \|\mathcal F_q(W+Z)-\mathcal F_q(W)-D\mathcal F_q(W)Z\|_{s'}
 &\le\frac{\mathcal C}{s-s'}\|Z\|_s^{3/2},\label{eq:endpoint_ckbounds}\\
 \|\mathcal F_q(0)\|_s&\le\frac{\mathcal K}{1-s}.\nonumber
\end{align}
The constants may be chosen to satisfy
\[
 \mathcal C\le M\rho^{-2}
       \left(\eta+\frac\delta\eta+\frac e{\delta\eta}\right),\qquad
 \mathcal K\le\frac{Me}{\delta\eta\rho^2}.
\]
For $W,W+Z$ in the unit ball, the unscaled nonlinear field is
$f(J,d,d_u)+Q(J,d)J_u$, with
$f(0,d,d_u)=O(e/\rho^2)$, $D_Jf=O(\delta+e/\rho^2)$,
and $Q=O(\delta+e/\rho)$. The component rescaling contributes
$\eta^{-1}$; a time derivative requires a smaller analytic disk by
Cauchy's estimate. The quadratic remainder implies the $3/2$ bound. The field is
holomorphic, so Nirenberg's theorem~\cite{nirenberg1972} applies with
the following radius estimate.

For $e>0$, choose $\eta=\eta_0\rho^2$, $\delta=\sqrt e$.
Then
\begin{equation}
 \mathcal C\le M\left(\eta_0+
            \frac{2\sqrt e}{\eta_0\rho^4}\right),\qquad
 \mathcal K\le\frac{M\sqrt e}{\eta_0\rho^4}.
 \label{eq:endpoint_scaledconstants}
\end{equation}
Choose $\eta_0$, then $c$, to make both constants small; $e=0$ gives
$R=1$. For an explicit radius, set $t=q/4$,
$\mathcal F_*(W,t)=4\mathcal F_{4t}(W)$,
$\mathcal C_*=4\mathcal C$, $\mathcal K_*=4\mathcal K$,
and $a_n=\tfrac34\prod_{j=2}^{n+1}(1-j^{-2})$.
Starting with $W_0=0$, let
\[
 G_n(t)=\int_0^t\mathcal F_*(W_n(\tau),\tau)\,d\tau-W_n(t),\qquad
 \lambda_n=\sup_{s,\,|t|<a_n(1-s)}
 \|G_n(t)\|_s\left(\frac{a_n(1-s)}{|t|}-1\right).
\]
Solve $Z_n=G_n+\int_0^tD\mathcal F_*(W_n)Z_n\,d\tau$ and put $W_{n+1}=W_n+Z_n$.
In the weighted norm defining $\lambda_n$, the Volterra operator
$Z\mapsto\int_0^tD\mathcal F_*(W_n)Z\,d\tau$ has norm at most
$4a_n\mathcal C_*$. To see this, apply Cauchy's estimate at the
intermediate radius $s(\tau)=(1+s-|\tau|/a_n)/2$ and integrate along
the radial segment from $0$ to $t$. Thus $8a_0\mathcal C_*<1$
bounds the weighted correction by $2\lambda_n$;
\eqref{eq:endpoint_ckbounds} then gives
$\lambda_0\le a_0\mathcal K_*$ and
$\lambda_{n+1}\le32\mathcal C_*\lambda_n^{3/2}(n+2)$.
Choose the constants in \eqref{eq:endpoint_scaledconstants} so that
\begin{equation}
 8a_0\mathcal C_*<1,\qquad 1024\mathcal C_*\sqrt{a_0\mathcal K_*}\le1,
 \qquad a_0\mathcal K_*<1/32.
 \label{eq:endpoint_newtonconstants}
\end{equation}
Then $\lambda_n\le a_0\mathcal K_*(n+1)^{-4}$ by induction, since
$(n+2)^5/(n+1)^6\le32$.
The total correction on successive smaller domains is at most
$2\sum_n\lambda_n(a_n/a_{n+1}-1)^{-1}<16a_0\mathcal K_*<1/2$.
Thus the iteration converges in the unit ball. Since $a_n\to3/8$,
$s=1/4$ gives $|q|<9/8$ and time radius $5\rho/16$. Cauchy estimates
give \eqref{eq:endpoint_jetbound} and its derivative bounds.

It remains to check that the independent unknowns are the derivatives
of one function $Y$. For
$C_1=U-Y_u$, $C_2=B_2-V_u$ and $C_3=A_2-U_u$,
\eqref{eq:endpoint_jets} gives
\[
 (C_1)_q=C_2,\qquad (C_3)_q=0,\qquad
 (C_2)_q=H_YC_1+H_UC_3+H_VC_2.
\]
Zero data give $C_1=C_2=C_3=0$, recovering the original equation.
Analytic Cauchy uniqueness and continuation give uniqueness on the
product domain; conjugation gives reality.
\end{proof}

\paragraph{A pulse with flat endpoints.}
For fixed $\gamma\ge3$, define in the rescaled units
\begin{align}
 \beta(u)&=\begin{cases}
  \exp[-1/(u(1-u))],&0<u<1,\\0,&u\notin(0,1),
 \end{cases}\qquad
 \propacc=\alpha\beta,\nonumber\\
 m(u)&=m_i\exp\!\left[-\gamma\int_0^u\propacc(v)\,\dd v\right].
 \label{eq:endpoint_bump}
\end{align}
These histories are smooth, with static ends, and obey
$-\dot m=\gamma m\propacc\ge3m|\propacc|$.
At $u_c\in(0,1)$, put $d_c=\min(u_c,1-u_c)$ and $\rho=d_c/8$.
On $|u-u_c|\le\rho$,
\[
 \Re\frac1{u(1-u)}=\Re\frac1u+\Re\frac1{1-u}
 \ge\frac{56}{81d_c}>\frac1{2d_c}.
\]
For $u_c\le1/2$, the first reciprocal is at least
$(7u_c/8)/(9u_c/8)^2$ and the second has positive real part; use $1-u$ for
the other half. The local analytic mass equation then gives
\begin{equation}
 \|m-m(u_c)\|_\rho+\|\propacc\|_\rho
       \le C\alpha e^{-1/(2d_c)}.
 \label{eq:endpoint_bumpbound}
\end{equation}
The constant is uniform for reference masses in a fixed compact mass interval.
Since $d_c^{-8}e^{-1/(2d_c)}$ is bounded, one sufficiently small $\alpha>0$
makes Lemma~\ref{lem:endpoint_analytic} apply at every $u_c$.

Analytic Cauchy uniqueness identifies solutions on overlapping time
disks, first at the equator and then by angular continuation.
For each $0<\epsilon<1/8$, Cauchy estimates give
\begin{equation}
 |\partial_u^k\partial_q^\ell(R-1)|
 \le C_{k\ell}\sqrt\alpha\,d_c^{-k}e^{-1/(4d_c)},\qquad |q|\le1+\epsilon.
 \label{eq:endpoint_flatjets}
\end{equation}
The exponential dominates every power of $d_c^{-1}$, so all mixed
derivatives of $R-1$ vanish at the endpoints. Extension by $R=1$ is smooth
and exact, with every fixed derivative norm tending to zero as
$\alpha\to0$.

For small $\alpha$, $E<0$, $\Delta<0$, $G>0$, and both metrics
remain Lorentzian. The embedding and curvatures approach the static
values uniformly on compact mass intervals below $12/25$, preserving
Type~I DEC. Gauss--Codazzi gives the flux balances.

\paragraph{Finite rapidity and a global cavity.}
For prescribed $d>0$, choose $I\Subset(0,12/25)$ containing
$[m_i e^{-\gamma d},m_i]$. The uniform smallness threshold gives an
allowed increment $d_0=\alpha_0\int_0^1\beta\,du>0$.
Concatenate finitely many translated pulses of increment at most $d_0$,
adjusting the last amplitude to sum to $d$. Vanishing endpoint derivatives allow
static gaps and smooth joining. Restoring $R_0$ uses
$\widetilde u=u/R_0$, $\widetilde m=m/R_0$, and
$\widetilde a=R_0a$, preserving $\int a\,du$.

The flat quotient on the simply connected time strip has a parallel
orthonormal coframe. Torsion freeness makes its forms closed, yielding
$b=-dT^2+dz^2$ with future orientation and
$T_u>0$, $J_{\rm dev}=\det\partial(T,z)/\partial(u,q)=\sqrt{-\Delta}>0$.
On each static end the coframe is constant, so $T\to\pm\infty$
uniformly in $q$. Each $T$ therefore determines a unique $u(T,q)$, and
\begin{equation}
 \left.\frac{\partial z}{\partial q}\right|_T
       =\frac{J_{\rm dev}}{T_u}>0.
 \label{eq:endpoint_injectivity}
\end{equation}
The injective quotient gives the spherical embedding
\begin{equation}
 (u,n)\longmapsto
 (T(u,n_3),R(u,n_3)n_1,R(u,n_3)n_2,z(u,n_3)),\qquad n\in S^2,
 \label{eq:endpoint_embedding}
\end{equation}
with rank three at the poles. Its constant-$T$ sections bound balls
by increasing height and positive cylindrical radius.

\paragraph{A temporal function across the junction.}
Since $b^{uu}=G/\Delta<0$, the inverse coordinate $u(T,z)$ has a
timelike gradient and is therefore a temporal function.
Extend it independently of cylindrical radius through the cavity.
In the exterior, let $\ell=r-R(u,q)\ge0$ and choose $L>0$.
The exact inverse metric gives
\begin{align}
 g^{-1}(\dd u,\dd\ell)&=-1,\nonumber\\
 B_\ell:=g^{-1}(\dd\ell,\dd\ell)
 &=1-\frac{2m}{r}+2\propacc rq+2R_u
   -2\propacc(1-q^2)R_q+\frac{(1-q^2)R_q^2}{r^2}.
 \label{eq:endpoint_normalnorm}
\end{align}
The finite construction has $R\ge R_{\min}>0$ and bounded histories and
shape derivatives. Therefore $(B_\ell)_+\le C_0+C_1\ell$ uniformly in time and
angle. Put
$\tau_+=u+\delta(1-e^{-\ell/L}),\quad p=e^{-\ell/L}/L>0$.
Then
\begin{equation}
 g^{-1}(\dd\tau_+,\dd\tau_+)
 =-2\delta p+\delta^2p^2B_\ell<0
 \quad\hbox{if}\quad
 \delta(C_0/L+C_1/\mathrm e)<2.
 \label{eq:endpoint_temporal}
\end{equation}
The orientation agrees with the future outgoing null vector $\partial_r$.
The temporal functions share the trace $u$. In Gaussian normal charts
the continuous metric gives a common convex timelike covector cone.
Local smoothing, with a locally finite partition and errors inside that
cone, gives a global smooth temporal function and stable
causality~\cite{minguzzi2019}.

\paragraph{The interior center and the mass law.}
The equatorial center $\Gamma$ lies inside each constant-$T$ cavity section.
Writing $s(u)=\sqrt{1-2m(u)/R_0}$, its proper time is
$\frac{\dd\tau_c}{\dd u}=\sqrt{s^2-(\propacc R_0)^2}>0$.
At the equator, $E=-s^2+(\propacc R_0)^2$, $F=\propacc R_0^2$,
$G=R_0^2$.
To compute the center's rapidity, use the orthonormal coframe
$\theta^0=h\,\dd u$, $\theta^1=j\,\dd q+k\,\dd u$, where
$h=\sqrt{-\Delta/G}$, $j=\sqrt G$, and $k=F/\sqrt G$.
The torsion equations give the boost connection
$\omega_u=[h_q+k(j_u-k_q)/h]/j$.
At $q=0$, $h=s$, $j=R_0$, $k=\propacc R_0$,
$j_u=0$, $k_q=-R_{qq}/R_0$, and
$h_q=\propacc(R_0-R_{qq})/s$.
Thus $\omega_u=\propacc/s$; the unknown angular curvature cancels.
The center tangent has coframe rapidity $\operatorname{artanh}(\propacc R_0/s)$, so
\begin{equation}
 \dot\chi_c=
 \frac{\dd}{\dd u}\operatorname{artanh}\frac{\propacc R_0}{s}
       +\frac{\propacc}{s}.
 \label{eq:endpoint_centerrate}
\end{equation}
The endpoint term vanishes. Since
$\dot s=\gamma\propacc(1-s^2)/(2s)$,
\begin{equation}
 \Delta\chi_c=\int\frac{\propacc}{s}\,\dd u
   =\frac2\gamma(\operatorname{artanh}s_f-
                         \operatorname{artanh}s_i),\qquad
 d=\int\propacc\,\dd u=\frac1\gamma\log\frac{m_i}{m_f}.
 \label{eq:endpoint_gains}
\end{equation}
Any finite positive center rapidity can also be selected:
$s_f=\tanh(\operatorname{artanh}s_i+\gamma\Delta\chi_c/2)<1$ specifies
a positive $m_f$ and hence a finite $d$.

\paragraph{Proof of Theorem~\ref{thm:turningjunction}.}
Choose a seed worldline $Z(u)$ with $v=\dot Z$ and a Fermi--Walker orthonormal triad $e_i$:
\[
 \dot v=A^ie_i,\qquad \dot e_i=A_i v,\qquad
 k=v+n^ie_i,\qquad x=Z+rk,\quad n\in S^2.
\]
The Minkowski pullback plus $2m\,\dd u^2/r$ is the arbitrary-direction photon rocket~\cite{podolsky2011},
\begin{align}
 \dd s_+^2={}&-\left(1-\frac{2m}{r}+2rA\cdot n
       -r^2[|A|^2-(A\cdot n)^2]\right)\dd u^2-2\dd u\dd r\nonumber\\
 &+2r^2\dd u\,A\cdot\dd n+r^2\dd\Omega^2,
 \qquad 4\pi n_{\rm rad}^2=-\dot m-3mA\cdot n.
 \label{eq:vectorrocket}
\end{align}
Here $n_{\rm rad}^2$ is the radiation coefficient, distinguished from the angular unit vector $n$.
For a nonzero edge, the initial unit tangent to its velocity geodesic is
$w_j=\frac{v_j-\cosh(d_j)v_{j-1}}{\sinh d_j}$.
Express $w_j$ in the current triad and use Theorem~\ref{thm:endpointburn} with $\gamma=3$ and that fixed acceleration direction.
Zero-length edges require no burn.
All pulse derivatives vanish on the intervening static collars, so the exterior metric and shell graph $r=R(u,n)$ are smooth.
The inertial ends ensure that every past light cone meets the seed; timelikeness makes the intersection unique.
Thus its retarded coordinates are globally injective away from the excised seed.

On a static collar the interior embedding is
\[
 Y=C_j+s_j(u-u_j)U_j^c+R_0E_{j,i}^c n^i,\qquad
 s_j=\sqrt{1-2m_j/R_0}.
\]
This matched frame generally differs from the seed frame. Identify
consecutive collars by a Poincar\'e transformation preserving angular
labels and static data. Insert a coast until a rest-time cap lies beyond
the preceding compact active pieces and static past. Attach the next
block above it: $h_{uu}<0$ keeps every generator on its future side.
The separating spacelike plane prevents overlap. Induction gives a
complete embedded boundary whose Minkowski-time sections bound balls.

Each block's interior temporal function $u$ is affine in rest time on a static collar, with slope $1/s_j$, so these functions patch.
For the exterior put $\ell=r-R(u,n)$ and $A_T=A-(A\cdot n)n$.
The inverse of \eqref{eq:vectorrocket} gives $g^{-1}(\dd u,\dd\ell)=-1$ and
\[
 B_\ell=1-\frac{2m}{r}+2rA\cdot n+2R_u
       -2A_T\cdot\nabla_{S^2}R+\frac{|\nabla_{S^2}R|^2}{r^2}.
\]
Finite bounded histories again give $(B_\ell)_+\le C_0+C_1\ell$.
Equation~\eqref{eq:endpoint_temporal} and the same cone smoothing therefore provide a global temporal function.
Strict surface DEC persists on each rotated block and static collar; \eqref{eq:vectorrocket} gives nonnegative radiation, and \eqref{eq:nonews} gives zero news for the full velocity history.
Integrating $-\dot m=3m|A|$ gives \eqref{eq:polygonmass}.
The velocity triangle inequality gives $D\le L$; equality in hyperbolic space requires successive segments of one geodesic with consistent orientation.
\hfill$\square$

\section{Global material reconstruction}
\label{app:material}

We fit the stress with a stored-energy law, impose material-radiation
exchange, and obtain the intrinsic wave speeds from the quadratic
perturbation energy.

The isolated timelike eigendirection of a strict Type~I stress defines
a smooth material velocity $U$. Flow from a spacelike sphere supplies
material labels and a clock $U\tau=1$.
A positive material area form $\omega$ gives
\begin{equation}
 J^a=\tfrac12\epsilon^{abc}\omega_{AB}D_bX^AD_cX^B=nU^a,
 \qquad n=|\omega_{12}|\sqrt{\det H},\qquad
 n_{H^{AB}}=\tfrac n2(H^{-1})_{AB}.
 \label{eq:materialcurrent}
\end{equation}
Closedness of the pulled-back area form gives $D_aJ^a=0$. Let $e_A^a$
be spatially dual to $D_aX^A$, so $h(e_A,e_B)=(H^{-1})_{AB}$.
On the selected history, denoted by overbars,
$\bar S_{ab}=\bar\varepsilon\bar U_a\bar U_b+\bar P_{ab}$ and
$\bar\gamma=\bar h+\bar U\otimes\bar U$.
Choose a constant $z_0>0$ such that, uniformly on the compact slab,
\begin{equation}
 0<\bar n z_0<\min(\bar\varepsilon,\bar\varepsilon+\bar P_1,
                                     \bar\varepsilon+\bar P_2).
 \label{eq:carriedmargin}
\end{equation}
For \eqref{eq:globalenergy}, take
\begin{equation}
 Q_{AB}=\bar e_A^a\bar e_B^b[\bar P_{ab}+(\bar\varepsilon-\bar nz_0)\bar\gamma_{ab}],
 \qquad V=-\tfrac12\bar\gamma^{ab}\bar P_{ab}.
 \label{eq:globalcoefficients}
\end{equation}
These smooth global tensors extend off the history with $Q>0$.
Metric variation of $I_m=-\int\mathcal E\,dV_h$, holding $(\tau,z)$ fixed, gives
\begin{equation}
 S_{ab}=Q_{AB}D_aX^AD_bX^B
 -(\tfrac12Q_{AB}H^{AB}+V)h_{ab}+nzU_aU_b.
 \label{eq:globalstress}
\end{equation}
At $(H,z)=(\bar H,z_0)$, tracing \eqref{eq:globalcoefficients} gives
$Q:\bar H=\bar P_1+\bar P_2+2(\bar\varepsilon-\bar nz_0)$, hence
$\mathcal E=\bar\varepsilon$ and $S=\bar S$.
Strict DEC persists by continuity.

For the expansion tensor $\Theta_{ab}$ of $U$,
$UH^{AB}=-2D^aX^AD^bX^B\Theta_{ab}$ and the chain rule gives
\begin{equation}
 U\mathcal E+(\mathcal E\gamma^{ab}+P^{ab})\Theta_{ab}
 =\mathcal E_\tau U\tau+nUz=-\mathcal L.
 \label{eq:globalfirstlaw}
\end{equation}
The last equality is \eqref{eq:globalreaction}.
On the history, junction energy balance gives $\mathcal E_\tau=-\bar{\mathcal L}$,
so $Uz=0$ at $z=z_0$.
Normalize the prescribed null direction by $-\bar U\cdot\bar k=1$ and write
$\bar k=\bar U+\bar v+\bar c\bar N$, $\bar c>0$.
Choose $d_A$ with $\bar v^a=d_AD^a\bar X^A$. In the state neighborhood set
\begin{equation}
 v^a=d_AD^aX^A,\quad c=\sqrt{1-v^2}>0,\quad
 k=U+v+cN,\quad T^{+\mu\nu}|_\Sh=(\mathcal L/c)k^\mu k^\nu.
 \label{eq:emissionclosure}
\end{equation}
For prescribed $\mathcal L=\bar{\mathcal L}\ge0$, this reproduces the
crossing tensor, with
$f_a=-T^+_{Na}=-\mathcal L(U_a+v_a)$ and $T^+_{NN}=\mathcal Lc$.
If $\mathfrak E_A$ is the label Euler--Lagrange expression, the diffeomorphism identity is
\begin{equation}
 D_bS^b{}_a=\mathfrak E_AD_aX^A-\mathcal E_\tau D_a\tau-nD_az.
 \label{eq:materialnoether}
\end{equation}
Set $\mathfrak E_A=e_A^a(f_a+\mathcal E_\tau D_a\tau+nD_az)$.
Equation~\eqref{eq:globalreaction} makes the parenthesis orthogonal
to $U$, giving $D_bS^b{}_a=f_a$. Off the history, normal balance
is a separate embedding equation.

To determine the characteristics, freeze a material rest frame with $H=I$.
For $\delta X=v\psi$, the rank-one spatial Jacobian has determinant
$\det(I+v\otimes\nabla\psi)=1+v\cdot\nabla\psi$.
Thus $nz$ has no quadratic spatial term; its timelike correction adds
$nz|v|^2(\partial_t\psi)^2/2$ to the kinetic energy.
The quadratic Lagrangian is
\begin{equation}
 \tfrac12v^T(Q+nzI)v\,(\partial_t\psi)^2
 -\tfrac12v^TQv\,|\nabla\psi|^2.
 \label{eq:materialquadratic}
\end{equation}
In general material coordinates its inertia is $M=Q+nzH^{-1}$, giving
\eqref{eq:materialspeed}.
For $\Pi=\partial_t\delta X$, $Y_i=\partial_i\delta X$, and
$\alpha=(\delta\tau,\delta z)$, the remaining principal terms are
\[
 M\partial_t\Pi=Q\partial_iY_i+B_i\partial_i\alpha,
 \qquad \partial_tY_i=\partial_i\Pi,\qquad \partial_t\alpha=0.
\]
The bounded change $Y_i\mapsto Y_i+Q^{-1}B_i\alpha$ separates the reaction variables.
The elastic energy $\Pi^TM\Pi+\sum_iY_i^TQY_i$ is positive and the
reaction block has zero rest-frame speed, with no additional Jordan block.

\subsection{Normal response and Sobolev growth}
\label{app:normalresponse}

\begin{lemma}[Persistence under a bounded normal response]
\label{lem:normalresponse}
Fix $\varepsilon,p>0$, put $c_*=\sqrt{p/\varepsilon}$, and let $k>0$.
Suppose elimination of the other perturbations is invertible near the closed disk
$|\sigma-c_*k|\le c_*k/2$ and gives
\[
 D(\sigma,k)=\varepsilon\sigma^2-pk^2+\mathcal F(\sigma,k).
\]
Assume $\mathcal F$ is holomorphic there and $|\mathcal F|\le C_0k$ on the disk, with $C_0\ge0$ independent of $k$.
If $k>4C_0/(3p)$, this disk contains exactly one zero of $D$, counted with multiplicity, and
\begin{equation}
 \operatorname{Re}\sigma>\tfrac12c_*k,\qquad
 |\sigma-c_*k|\le\frac{2C_0}{3\varepsilon c_*}.
 \label{eq:responsegrowth}
\end{equation}
\end{lemma}
\begin{proof}
On the boundary, $|\sigma+c_*k|\ge3c_*k/2$, so
\[
 |\varepsilon\sigma^2-pk^2|
 =\varepsilon|\sigma-c_*k|\,|\sigma+c_*k|
 \ge\tfrac34pk^2>C_0k.
\]
Rouch\'e's theorem gives the same zero count as the quadratic, namely one.
The disk lies in $\operatorname{Re}\sigma\ge c_*k/2$, and the zero is interior.
At that zero, the same factorization and
$|\sigma+c_*k|\ge3c_*k/2$ give the second bound.
\end{proof}

The radiating junction still requires constrained invertibility on the
complex disk and reconstruction of a physical mode. The vacuum response
in~\cite{yang2023} does not establish these properties.

\begin{corollary}[Failure of bounded frozen normal evolution]
\label{cor:frozenillposed}
Let $\varepsilon>0$ and at least one principal pressure be positive.
On a periodic tangent plane, \eqref{eq:normalprincipal} has no bounded solution map
$H^s\times H^{s-1}\to H^{s-r}$ at any fixed positive time, for any $s\in\mathbb R$ and finite $r\ge0$.
The full frozen surface system in Proposition~\ref{prop:membrane} likewise has no bounded Sobolev evolution with any finite derivative loss for its first-order variables.
\end{corollary}
\begin{proof}
The open cone of positive-pressure directions contains a reciprocal-lattice direction. Choose its integer multiples as wave vectors and put $c=\sqrt{p(\hat k)/\varepsilon}>0$.
A Fourier mode has the growing solution $\xi_k=A_ke^{c|k|t}e^{ikx}$.
For a prescribed $t_*>0$, take $A_k=|k|^{-s}e^{-c|k|t_*/2}$.
The initial $H^s\times H^{s-1}$ norm tends to zero, whereas the $H^{s-r}$ displacement norm at $t_*$ grows as
$|k|^{-r}e^{c|k|t_*/2}\to\infty$.
The full system uses its growing unit eigenvectors and
$\operatorname{Re}\sigma=c|k|+O(1)$. Real or imaginary parts yield
real perturbations with the same growth.
\end{proof}

\section{Radiating free-boundary estimates}
\label{app:transport}

Conservation propagates the Einstein constraints. We then estimate
photon transport across moving faces and use transverse crossings
to control normal derivatives in the coupled evolution.
Here $C^{1,1}$ means continuously differentiable with Lipschitz first derivatives.

\subsection{Constraints and photon transport}

Static initial data are compatible with delayed emission. For
$g=-e^{2U}dt^2+e^{-2U}h_{ij}dx^idx^j$, spatial harmonic coordinates
are spacetime harmonic: $\Box_gt=0$ and
$\Box_gx^i=e^{2U}\Delta_hx^i=0$.
Here $U$ denotes the static metric potential.

\begin{theorem}[Constraint propagation with material-radiation exchange]
\label{thm:bulkconstraints}
Let a reduced solution have a piecewise smooth $C^{1,1}$ metric,
a smooth timelike material boundary, complete uniformly spacelike slices,
and bounded wave-energy coefficients, lapse, and inverse lapse.
Assume total stress $T=\mathbf1_{\Omega_t}T_{\rm mat}+T_\gamma$
obeys the exchange law and $N_\mu T_{\rm mat}^{\mu\nu}=0$ at the faces.
If the harmonic residual $C^\mu=g^{\alpha\beta}\Gamma^\mu_{\alpha\beta}$
is in the wave-energy class and
\begin{equation}
 G_{\mu\nu}-8\pi GT_{\mu\nu}
 =\nabla_{(\mu}C_{\nu)}-\tfrac12g_{\mu\nu}\nabla_\alpha C^\alpha,
 \label{eq:reducedresidual}
\end{equation}
then initial $C=0$ and the normal Einstein constraints imply $C=0$ throughout.
\end{theorem}
\begin{proof}
Traction removes the surface divergence; volume exchange cancels.
The Bianchi identity gives, allowing an $L^2$ conservation error
$Q_\nu=\nabla^\mu T_{\mu\nu}$ with no surface measure,
\begin{equation}
 \Box_g C_\nu+R_\nu{}^\mu C_\mu=-16\pi GQ_\nu.
 \label{eq:subsidiary}
\end{equation}
Multiplication by $\partial_tC$, cancellation of the conormal interface
fluxes, and Gronwall give
\begin{equation}
 E_C(t)\le e^{Kt}\left[E_C(0)+K\int_0^t e^{-Kr}\|16\pi GQ(r)\|_2^2dr\right],
 \quad E_C\simeq\|C\|_{H^1}^2+\|\partial_tC\|_2^2.
 \label{eq:constraintenergy}
\end{equation}
Initially $C=0$ removes its tangential derivatives; the normal components
of \eqref{eq:reducedresidual} give $\nabla_0C_\mu=0$ in an adapted
orthonormal frame. Thus $E_C(0)=0$, and conservation gives $Q=0$.
\end{proof}

\begin{proposition}[Photon estimates across a material interface]
\label{prop:photonenergy}
On a finite slab, prescribe a $C^{1,1}$ metric, possibly piecewise smooth
across a moving timelike interface, with uniform coefficient bounds,
positive lapse and spatial metric, and bounded inverses. In canonical
momenta $p_i$, let the initial data and $S=\mathcal C/p^0$ be supported
in $0<p_-\le|p|\le p_+<\infty$.
The massless transport problem has a unique $L^2$ solution for
$f_0\in L^2$, $S\in L^1_tL^2$, preserving positivity, with
\begin{align}
 \|f(t)\|_2&\le\|f_0\|_2+\int_0^t\|S(r)\|_2dr,\nonumber\\
 \|f(t)\|_{H^1}&\le e^{Ct}\left[\|f_0\|_{H^1}
                       +\int_0^t\|S(r)\|_{H^1}dr\right].
 \label{eq:photonenergy}
\end{align}
The second estimate requires the indicated $H^1$ data. Momentum support
stays in $p_-e^{-Ct}\le|p|\le p_+e^{Ct}$, and
$\|T_\gamma(t)\|_{L^2_x}\le C_T\|f(t)\|_2$, without an angular
transversality assumption.
\end{proposition}
\begin{proof}
For $g=-N^2dt^2+h_{ij}(dx^i+\beta^idt)(dx^j+\beta^jdt)$,
\[
 \mathcal H=-p_0=N\sqrt{h^{ij}p_ip_j}-\beta^ip_i,\qquad
 V=(\partial_p\mathcal H,-\partial_x\mathcal H).
\]
On compact momentum annuli $V$ is Lipschitz, divergence-free, and
$|\dot p|\le C|p|$. Its flow preserves $dx\,dp$, so integration of
$f_t+V\cdot Df=S$ gives positivity and the first estimate.
Differentiation gives
$(\partial_t+V\cdot D)Df=DS-(DV)^TDf$; the $L^\infty$ bound on
$DV$ gives the second by approximation.
The moment kernel $p^\mu p^\nu/(\sqrt{-g}\,p^0)$ is bounded on the
propagated annulus, giving the stress estimate by Cauchy--Schwarz.
\end{proof}
The same argument for differences gives
\[
 \|\delta f(t)\|_2\le\|\delta f_0\|_2+
 \int_0^t(\|\delta S\|_2+\|\delta V\|_\infty\|Df_2\|_2)dr,
 \qquad \|\delta V\|_\infty\le C\|\delta g\|_{W^{1,\infty}}.
\]
A source nonzero at a face, extended by zero, need not be $H^1$.

\begin{proposition}[Loss of regularity at grazing rays]
\label{prop:grazing}
A nonnegative source on a static Minkowski annulus, smooth up to its faces and extended by zero, can produce $f\notin H^1_{\rm loc}(x,p)$ from zero data and an energy density that is not piecewise $H^2$. The source may have compact time and positive-momentum support. Vanishing to order $k\ge0$ at the outer face can still leave $f\notin H^{k+1}_{\rm loc}$ downstream.
\end{proposition}
\begin{proof}
For $\omega=p/|p|$, put $\ell=x\cdot\omega$ and
$b=|x-\ell\omega|$. Choose
$S=a(t)\psi(p)\mathbf1_{\{R_-<|x|<R_+\}}$, with smooth nonnegative
$a,\psi$, $\psi>0$ on an open momentum set, and $a=1$ on a long
finite interval. On an open downstream set whose full crossing occurs
during that plateau, rays near $b=R_+$ miss the inner face and give
\begin{equation}
 f=2\psi(p)\sqrt{R_+^2-b^2}\,\mathbf1_{\{b<R_+\}}.
 \label{eq:grazingchord}
\end{equation}
Its squared $b$ derivative has a logarithmically divergent integral
against $b\,db\,d\ell\,d\phi$ on an open momentum set.
For radial $\psi$, the energy density during the plateau is proportional
to $I_{R_+}(r)-I_{R_-}(r)$, where polar integration gives
\begin{equation}
 I_R(r)=\int_{|y|<R}\frac{d^3y}{|x-y|^2}
 =\frac{2\pi}{r}\int_0^R s\log\frac{r+s}{|r-s|}\,ds
 =2\pi R+\pi\frac{R^2-r^2}{r}\log\frac{r+R}{|r-R|}.
 \label{eq:grazingmoment}
\end{equation}
At $r=R+\delta$ the singular term is $2\pi\delta\log|\delta|$,
with second derivative $2\pi/\delta\notin L^2$ on either side.
The inner-ball term is smooth near $R_+$.
Multiplying the source by $(R_+^2-|x|^2)^k$ replaces its chord by
\[
 \psi(p)\int_{-\sqrt d}^{\sqrt d}(d-\ell^2)^k d\ell
 =\psi(p)B(1/2,k+1)d^{k+1/2},\qquad d=R_+^2-b^2.
\]
Its first $k$ weak transverse derivatives are square integrable,
but its $(k+1)$st is proportional to $d^{-1/2}$.
\end{proof}

\subsection{Transverse photon coupling}
\label{app:transverse}

Material coordinates fix the moving faces. We estimate derivatives
tangent to those faces first, then recover normal derivatives from
transport by dividing by the nonzero crossing velocity.

\begin{proof}[Proof of Theorem~\ref{thm:transversebody}]
Extend the material map to a spatial diffeomorphism
$x=\Phi(t,X)$ of a larger fixed domain.
Use the Cartesian metric components $\widetilde g=g\circ\Phi$ and retain
the Eulerian covariant photon momenta $p_i$.
Writing $J=\det D_X\Phi$, the phase velocity is
\begin{equation}
 W^X=(D_X\Phi)^{-1}(\partial_p\mathcal H-\Phi_t),
 \qquad W^p=-(\partial_x\mathcal H)\circ\Phi,
 \qquad
 \partial_tJ+\operatorname{div}_{X,p}(JW)=0.
 \label{eq:materialphase}
\end{equation}
Evaluate the Hamiltonian at $(t,\Phi(t,X),p)$; these coefficients use
only first derivatives of $(\widetilde g,\Phi)$.
In fixed face charts define
$\|F\|_{\mathcal P_s}=\sum_{j=0}^s
\|\partial_t^jF\|_{H^{s-j}_{X,p}(\mathrm{pieces})}$ on the body,
cavity, and exterior, with $\mathcal P_s^x$ omitting momentum variables.
These are instantaneous norms; time-uniform bounds use their supremum
over $[0,T]$. The momentum integrals are over the compact annulus
propagated in Proposition~\ref{prop:photonenergy}.
For the metric, use the transmission spaces of~\cite{andersson2014}:
its $j$th time derivative is piecewise $H^{s+1-j}$ and globally
$H^{\min(2,s+1-j)}$, for $0\le j\le s+1$.
Common first metric traces make $W$ Lipschitz. Near the small-$G$
static state, short-time characteristic continuity gives
$W^X\cdot X/|X|>1/2$ near all emitted rays.

Transversality gives regular crossing times and $[f]=0$. In a fixed
face chart, all time, momentum, and face-tangential derivatives of this
trace agree. The normal derivative need not agree: transport gives
\begin{equation}
 W^n[\partial_n f]=[S],\qquad
 S=(\mathcal C/p^0)\circ\Phi,
 \label{eq:photontracejump}
\end{equation}
where the source is zero-extended and its initial time jets vanish.
Thus the appropriate high-order spaces are piecewise Sobolev spaces.
Let $\mathcal T_s(f)$ be the $L^2(J\,dX\,dp)$ sum of derivatives through
order $s$ in $t$, $p$, and directions tangent to the fixed faces, together
with all derivatives away from their collars.
The interface flux for each such derivative $D_\parallel^\alpha f$ is
$[JW^n|D_\parallel^\alpha f|^2]$, which vanishes by the common traces
of $f$, $J$, and $W$.
The commuted equation has source
$D_\parallel^\alpha S-[D_\parallel^\alpha,W\cdot D]f$.
Every commutator term contains at most $s$ derivatives of $W$ and
at most $s$ derivatives of $f$; the top derivative of $W$ multiplies
only $Df$. Sobolev products in the six phase-space variables therefore
close for $s\ge8$.
Normal derivatives are recovered from
\begin{equation}
 \partial_n f=
 \frac{S-\partial_tf-W_{\rm tan}\cdot D_{\rm tan}f-W^p\cdot D_pf}{W^n},
 \qquad |W^n|\ge c>0.
 \label{eq:normalrecovery}
\end{equation}
Localize near emitted rays, where $|W^n|\ge c$. Each use of
\eqref{eq:normalrecovery} replaces a normal derivative by tangential,
time, or momentum derivatives and a source derivative. Induction gives
\begin{align*}
 \|f\|_{\mathcal P_s}
 &\le C_R\bigl(\mathcal T_s(f)+\|S\|_{\mathcal P_{s-1}}\bigr),\\
 \frac{d}{dt}\mathcal T_s(f)
 &\le C_R\bigl(\|f\|_{\mathcal P_s}+\|S\|_{\mathcal P_s}\bigr).
\end{align*}
The first bound controls the normal derivatives in the commutators of
the second. Coefficients use at most $s+1$ derivatives of
$(\widetilde g,\Phi)$. Since the initial source jets vanish,
$\|S(t)\|_{\mathcal P_{s-1}}\le\int_0^t\|S(r)\|_{\mathcal P_s}\,dr$.
Gronwall's inequality gives
\begin{equation}
 \|f(t)\|_{\mathcal P_s}
 \le C_R\int_0^t\|S(r)\|_{\mathcal P_s}\,dr,
 \qquad
 \|T_\gamma(t)\|_{\mathcal P_s^x}\le C_R\|f(t)\|_{\mathcal P_s}.
 \label{eq:transverseestimate}
\end{equation}
Here $s\ge8$ controls phase-space products; constants depend on the
state bounds, $T,c$. The difference source
$\delta S-\delta W\cdot Df_2$ closes one order lower. A common trace
approximation preserves the interface cancellation.

Freeze the material-coordinate metric $h$, then solve the elastic
Neumann problem with the prescribed fuel and recoil using
\cite[Theorem 4.3]{andersson2014}. Positive inertia and Korn coercivity
apply; the fuel term has no spatial deformation gradient and adds no
traction. Solve transport next, then the linear metric transmission
problem with both stresses using \cite[Theorem 4.2]{andersson2014}.
The source space $\mathcal X_T^s$ in that wave theorem requires
piecewise spatial $H^{s-j}$ for $j$ time derivatives, with only global
$L^2$ regularity. This differs from the metric transmission space and
allows different normal derivatives on the two sides of a face.
The compact momentum support and smooth moment kernel give
\[
 \|\partial_t^j T_\gamma\|_{H^{s-j}(\mathrm{pieces})}
 \le C_R\sum_{\ell=0}^j
       \|\partial_t^\ell f\|_{H^{s-\ell}_{X,p}(\mathrm{pieces})}.
\]
Thus \eqref{eq:transverseestimate} supplies the wave source without
requiring $f$ or $T_\gamma$ to be globally $H^2$.
Delayed emission keeps all initial metric and traction jets static.
On a fixed finite domain containing the causal disturbance, small
$T,\eta$ preserve a closed bounded metric ball and the strict coefficient
margins. Space-time compactness two orders lower and the difference
estimates give a continuous compact map. Schauder gives a fixed point;
weak lower semicontinuity preserves the higher bounds.

For differences, let $D_s$ be the metric-material energy of order $s$
(including velocity at order $s-1$), and
$F_{s-1}=\|\delta f\|_{\mathcal P_{s-1}}$.
Elliptic recovery and the first physical metric derivatives, as in
\cite[section 5.1.2]{andersson2014}, give
\[
 D_s(t)\le C_R\int_0^t(D_s+F_{s-1})dr,\qquad
 F_{s-1}(t)\le C_R\int_0^tD_sdr.
\]
Photon moments enter the wave source, while recoil is independent of
$f$; Gronwall proves uniqueness. Conservation restores the Einstein
constraints. Characteristic continuity keeps photons outside the cavity,
and domain of dependence joins to the asymptotically flat exterior.
\end{proof}

\subsection{Momentum control with self-gravity}
\label{app:controlG}

We select the target momentum by inverting the leading linear impulse.
Uniform remainder bounds make this inversion a contraction.

\begin{proof}[Proof of Corollary~\ref{cor:gravitatingcontrol}]
The static family and the estimates of Theorem~\ref{thm:transversebody}
are uniform for small $G$. Choose $T$ before the pulse and then reduce
$\eta$. Write $\epsilon_G=GM_0/R_+$, where
$M_0=\int_{\mathcal B}(m+z_0)\,dX$; the controls satisfy $|b_*|<5/16$.

Subtract the static equations in material coordinates. The elastic
energy estimate, the photon estimate \eqref{eq:transverseestimate}, and
the metric transmission estimate give, on $[0,T]$,
\[
 \|\Phi-\Phi_G\|_{s+1}+\|f\|_{\mathcal P_s}\le C\eta,
 \qquad \|\widetilde g-\widetilde g_G\|_{s+1}\le C\epsilon_G\eta.
\]
These are the mixed transmission norms of that theorem.
The metric perturbation has zero initial data and sources proportional
to $G$. Control differences satisfy the same estimates one order lower,
with an extra factor $|\theta-\vartheta|$, since fuel is common and
recoil is affine in $b_*$. The elastic difference flux has the form
$A^{\alpha\beta}\partial_\beta\delta\Phi+B^\alpha\delta\widetilde g$;
its full normal component vanishes at the faces. Keeping the metric
term in that conormal condition permits the difference estimate of
\cite[section 5.1.2]{andersson2014}. One higher construction order
supplies the coefficient bounds.

The material exchange equation and zero traction give the exact identity
\begin{equation}
 \dot P_i=\frac12\int_{\Omega_t}\sqrt{-\det g}\,
 T_{\rm mat}^{\mu\nu}\partial_i g_{\mu\nu}\,d^3x
 -\eta\int_{\mathcal B}a\chi z\,(U_i+d_{b,i})\,dX.
 \label{eq:gravitatingimpulse}
\end{equation}
For the second term use $\sqrt{-\det g}\,nU^0\det D\Phi=1$.
The gravitational integral vanishes on the static reference; its
perturbation and control differences are $O(G\eta)$ and
$O(G\eta|\theta-\vartheta|)$, with the reference scales fixed.
At $G=\eta=0$, $e_r=X/|X|$, and radial integration gives
\[
 \int_{\mathcal B}\chi z_0e_r\,dX=0,\qquad
 \int_{\mathcal B}\chi z_0\mathsf C_{e_r}\,dX
       =\frac{1-\beta^2}{3}I_zI.
\]
Since the rotating part of \eqref{eq:conecontrol} has zero $a$-weighted mean,
\begin{align}
 P(T;\theta)&=-\eta K_c\theta+R_{G,\eta}(\theta),\nonumber\\
 |R_{G,\eta}(\theta)|&\le C\eta(\epsilon_G+\eta),\qquad
 |R_{G,\eta}(\theta)-R_{G,\eta}(\vartheta)|
 \le C\eta(\epsilon_G+\eta)|\theta-\vartheta|.
 \label{eq:gravitatingremainder}
\end{align}
Choose $C(\epsilon_G+\eta)\le K_c\min(r_c/2,1/2)$.
Then $\theta\mapsto[R_{G,\eta}(\theta)-p_*]/(\eta K_c)$ maps the
radius-$r_c$ ball into itself with contraction constant at most $1/2$.
It gives the claimed unique control.
At phases $\pi/2$ and $\pi$, the leading force directions have cross
product of magnitude at least $b_c^2-2b_cr_c=b_c^2/2$.
The pointwise version of \eqref{eq:gravitatingimpulse} preserves this
margin for small $\epsilon_G+\eta$; fuel depletion is independent of $\theta$.
\end{proof}

\acknowledgments
This work is financially supported by VinUniversity under the Environmental
Intelligence (CEI) Grant (No.~VUNI.CEI.FS 0009).


\begin{thebibliography}{99}
\newcommand{\eprint}[1]{\href{https://arxiv.org/abs/#1}{arXiv:#1}}

\bibitem{alcubierre1994}
M.~Alcubierre, \emph{The warp drive: hyper-fast travel within general relativity},
\emph{Class. Quantum Grav.} \textbf{11} (1994) L73 [\eprint{gr-qc/0009013}].

\bibitem{natario2002}
J.~Nat\'ario, \emph{Warp drive with zero expansion},
\emph{Class. Quantum Grav.} \textbf{19} (2002) 1157 [\eprint{gr-qc/0110086}].

\bibitem{santiago2022}
J.~Santiago, S.~Schuster and M.~Visser, \emph{Generic warp drives violate the null energy condition},
\emph{Phys. Rev. D} \textbf{105} (2022) 064038 [\eprint{2105.03079}].

\bibitem{bobrick2021}
A.~Bobrick and G.~Martire, \emph{Introducing physical warp drives},
\emph{Class. Quantum Grav.} \textbf{38} (2021) 105009 [\eprint{2102.06824}].

\bibitem{fuchs2024}
J.~Fuchs, C.~Helmerich, A.~Bobrick, L.~Sellers, B.~Melcher and G.~Martire,
\emph{Constant velocity physical warp drive solution},
\emph{Class. Quantum Grav.} \textbf{41} (2024) 095013 [\eprint{2405.02709}].

\bibitem{membrane2023}
G.~Huey, \emph{Membrane models as a means of propulsion in general relativity:
super-luminal warp-drive that satisfies the weak energy condition},
\emph{Class. Quantum Grav.} \textbf{41} (2024) 135007 [\eprint{2311.07193}].

\bibitem{barzegar2026}
H.~Barzegar, T.~Buchert and Q.~Vigneron,
\emph{General formalism, classification, and demystification of the current warp-drive spacetimes},
\eprint{2602.16495}.

\bibitem{buchert2026frackowiak}
T.~Buchert and A.~Frackowiak,
\emph{Novel realizations of warp drive spacetimes as solutions of general relativity},
\emph{Universe} \textbf{12} (2026) 132 [\eprint{2605.03653}].

\bibitem{kinnersley1969}
W.~Kinnersley, \emph{Field of an arbitrarily accelerating point mass},
\emph{Phys. Rev.} \textbf{186} (1969) 1335.

\bibitem{bonnor1994}
W.~B.~Bonnor, \emph{The photon rocket},
\emph{Class. Quantum Grav.} \textbf{11} (1994) 2007.

\bibitem{damour1995}
T.~Damour, \emph{Photon rockets and gravitational radiation},
\emph{Class. Quantum Grav.} \textbf{12} (1995) 725 [\eprint{gr-qc/9412063}].

\bibitem{dain1996}
S.~Dain, O.~M.~Moreschi and R.~J.~Gleiser,
\emph{Photon rockets and the Robinson--Trautman geometries},
\emph{Class. Quantum Grav.} \textbf{13} (1996) 1155 [\eprint{gr-qc/0203064}].

\bibitem{podolsky2011}
J.~Podolsk\'y, \emph{Photon rockets moving arbitrarily in any dimension},
\emph{Int. J. Mod. Phys. D} \textbf{20} (2011) 335 [\eprint{1006.1583}].

\bibitem{fernandez2026}
F.~Fern\'andez-\'Alvarez and J.~M.~M.~Senovilla,
\emph{Asymptotic gravitational radiation of photon rockets with $\Lambda$:
point, string, and sheet sources}, \eprint{2608.31129}.

\bibitem{fuzfa2019}
A.~F\"uzfa, \emph{Interstellar travels aboard radiation-powered rockets},
\emph{Phys. Rev. D} \textbf{99} (2019) 104081 [\eprint{1902.03869}].

\bibitem{henriques2012}
P.~G.~Henriques and J.~Nat\'ario, \emph{The rocket problem in general relativity},
\emph{J. Optim. Theory Appl.} \textbf{154} (2012) 500 [\eprint{1105.5235}].

\bibitem{yang2023}
H.~Yang, B.~Bonga and Z.~Pan, \emph{Dynamical instability of self-gravitating membranes},
\emph{Phys. Rev. Lett.} \textbf{130} (2023) 011402 [\eprint{2207.13754}].

\bibitem{mourao2025}
P.~Mour\~ao, J.~Nat\'ario and R.~Vicente,
\emph{Relativistic elastic membranes: rotating disks and Dyson spheres},
\emph{Class. Quantum Grav.} \textbf{42} (2025) 025023 [\eprint{2409.10602}].

\bibitem{pitre2026}
T.~Pitre, B.~Schneider and E.~Poisson,
\emph{Self-gravitating thin shells are dynamically unstable on all angular scales},
\eprint{2604.05980}.

\bibitem{bondi1962}
H.~Bondi, M.~G.~J.~van der Burg and A.~W.~K.~Metzner,
\emph{Gravitational waves in general relativity. VII. Waves from axi-symmetric isolated systems},
\emph{Proc. R. Soc. Lond. A} \textbf{269} (1962) 21.

\bibitem{sachs1962}
R.~K.~Sachs, \emph{Gravitational waves in general relativity. VIII. Waves in asymptotically flat space-time},
\emph{Proc. R. Soc. Lond. A} \textbf{270} (1962) 103.

\bibitem{israel1966}
W.~Israel, \emph{Singular hypersurfaces and thin shells in general relativity},
\emph{Nuovo Cimento B} \textbf{44} (1966) 1;
erratum \textbf{48} (1967) 463.

\bibitem{poisson2004}
E.~Poisson, \emph{A Relativist's Toolkit: The Mathematics of Black-Hole Mechanics},
Cambridge University Press (2004).

\bibitem{hawking1973}
S.~W.~Hawking and G.~F.~R.~Ellis, \emph{The Large Scale Structure of Space-Time},
Cambridge University Press (1973).

\bibitem{horvatilijic2007}
D.~Horvat and S.~Iliji\'c, \emph{Gravastar energy conditions revisited},
\emph{Class. Quantum Grav.} \textbf{24} (2007) 5637 [\eprint{0707.1636}].

\bibitem{andreasson2008sharp}
H.~Andr\'easson, \emph{Sharp bounds on $2m/r$ of general spherically symmetric static objects},
\emph{J. Differential Equations} \textbf{245} (2008) 2243 [\eprint{gr-qc/0702137}].

\bibitem{teschl2012}
G.~Teschl, \emph{Ordinary Differential Equations and Dynamical Systems},
Graduate Studies in Mathematics \textbf{140}, American Mathematical Society,
Providence, RI (2012).

\bibitem{nirenberg1972}
L.~Nirenberg, \emph{An abstract form of the nonlinear Cauchy--Kowalewski theorem},
\emph{J. Differential Geom.} \textbf{6} (1972) 561.

\bibitem{minguzzi2019}
E.~Minguzzi, \emph{Lorentzian causality theory},
\emph{Living Rev. Relativ.} \textbf{22} (2019) 3.

\bibitem{balart2014}
L.~Balart and E.~C.~Vagenas,
\emph{Regular black hole metrics and the weak energy condition},
\emph{Phys. Lett. B} \textbf{730} (2014) 14 [\eprint{1401.2136}].

\bibitem{andersson2014}
L.~Andersson, T.~A.~Oliynyk and B.~G.~Schmidt,
\emph{Dynamical compact elastic bodies in general relativity},
\emph{Arch. Ration. Mech. Anal.} \textbf{220} (2016) 849 [\eprint{1410.4894}].

\bibitem{andersson2007}
L.~Andersson, R.~Beig and B.~G.~Schmidt,
\emph{Static self-gravitating elastic bodies in Einstein gravity},
\emph{Commun. Pure Appl. Math.} \textbf{61} (2008) 988 [\eprint{gr-qc/0611108}].

\bibitem{hu2010}
Z.~Hu, \emph{A modern view of the classical Herglotz--Noether theorem},
\eprint{1004.1935}.

\end{thebibliography}
\end{document}